\documentclass[11pt]{article}

\usepackage[margin=1in]{geometry}
\usepackage{amsmath,amssymb,amsthm,mathtools}
\usepackage{mathrsfs}
\usepackage{booktabs}
\usepackage{multirow}
\usepackage{graphicx}
\usepackage{tikz}
\usetikzlibrary{arrows.meta,positioning,calc}
\usepackage{enumitem}
\usepackage[expansion=false]{microtype}
\usepackage[colorlinks=true,linkcolor=blue!60!black,citecolor=blue!60!black,urlcolor=blue!60!black]{hyperref}
\allowdisplaybreaks

\newtheorem{theorem}{Theorem}
\newtheorem{lemma}{Lemma}
\newtheorem{proposition}{Proposition}
\newtheorem{corollary}{Corollary}
\newtheorem{definition}{Definition}

\newtheorem{remark}{Remark}
\newtheorem{axiom}{Axiom}

\DeclareMathOperator{\CVaR}{CVaR}

\newcommand{\Ex}{\mathbb{E}}
\newcommand{\Prob}{\mathbb{P}}
\newcommand{\KK}{\mathcal{K}}
\newcommand{\TT}{\mathcal{T}}
\newcommand{\CC}{\mathscr{C}}
\newcommand{\Cone}{\mathcal{C}}
\newcommand{\Stil}[1]{\widetilde{S}_{#1}}
\newcommand{\bits}[1]{2^{#1}}
\newcommand{\pS}{\textsuperscript{\tiny S}}
\newcommand{\pD}{\textsuperscript{\tiny D}}
\newcommand{\pF}{\textsuperscript{\tiny F}}

\newcommand{\succone}{\succeq_{\Cone}}

\title{\textbf{GAUGE: A Formal Framework for Measuring Cryptographic
Security under Heterogeneous Adversary Cost Models}}

\author{
Bhanwar Gupta\thanks{Corresponding author. Email: \texttt{bgupta55@gmail.com}. ORCID: 0009-0009-1524-6103.}
\quad\and\quad
Sanjeev Rana\thanks{Email: \texttt{dr.sanjeevrana@mmumullana.org}. ORCID: 0000-0003-2911-3052.}
\\[1ex]
\normalsize Department of Computer Science and Engineering,\\
\normalsize Maharishi Markandeshwar Engineering College,\\
\normalsize Maharishi Markandeshwar (Deemed to be University),\\
\normalsize Mullana, Ambala, Haryana 133207, India
}
\date{}

\begin{document}
\maketitle

\begin{abstract}
Standards bodies, vendors, and migration planners report the security of
a cryptographic scheme as a single number of bits. That number is not a
property of the scheme: it is the output of an adversary cost model, an
accounting convention for what an attacker's time, memory, and quantum
resources are worth, and different bodies adopt different conventions.
NIST prices memory into its comparisons; ANSSI and BSI do not. The two choices can,
and provably do, disagree about which of two standardized schemes is more
secure.

GAUGE represents the security of a scheme not as a scalar but as a
function on the space of admissible cost models: a security profile.
Comparisons between schemes become comparisons between these functions.
A scheme that is ranked higher under one accounting convention and lower
under another is not a measurement error; it is a crossing profile, a
geometric fact about the two functions that any scalar rating necessarily
discards.

We formalize price functionals over a cone of adversary cost models, show
that resulting security profiles are piecewise-linear and concave, and prove that whenever two profiles cross,
no rating that is simultaneously faithful to the underlying costs, total
over all comparable pairs, and independent of the particular cost model
can exist. We call this the rating trilemma. We complement it with a
polynomial-time linear-programming procedure that certifies, for any
pair of schemes, whether their ranking is robust across all admissible
models, reverses under some models, or is genuinely incomparable.

We extend GAUGE with a two-layer risk measure: a coherent risk functional
aggregates stochastic cryptanalytic decay given a cost model, and a
credal-set construction aggregates systematic uncertainty about which
cost model is appropriate.

We instantiate GAUGE on the NIST post-quantum standards (ML-KEM,
ML-DSA-adjacent families, and classical anchors) and on a curated
25-year chronology of cryptanalytic breaks. The linear program certifies
a rating reversal for ML-KEM-512 against its AES-128 anchor that a
$4$-$5\%$ shift in how memory is priced is sufficient to trigger, we
measure a lattice-sieving cost drift of $9.79$ bits per year over an
eight-year window, and we show that a hybrid X25519 + ML-KEM-768
handshake reduces combined-break probability twenty-fold at a $2.3$
kilobyte cost. The artifact reproduces every table and figure in under
seven seconds. The result is a reporting format that keeps the accounting convention
explicit, and that gives standards bodies a linear-programming certificate
for when a disagreement between two cost-model positions is real.
\end{abstract}


\bigskip
\noindent\textit{Use of AI Tools: An AI assistant was used in formatting
and editing of the manuscript; all theorems, results, and numbers were
executed and verified by the authors, and no results were fabricated.}

\section{Introduction}\label{sec:intro}

\subsection{Scalar security estimates and their hidden conditioning}

Parameter selection, hybrid protocol design, migration deadlines, and
standards categories all take a scalar ``bits of security'' as input.
That scalar is not a property of a cryptographic scheme alone: it
depends on the adversary cost model under which the underlying attack
cost was estimated, and that model is almost never stated.

Table~\ref{tab:anatomy} makes this concrete for a single standardized
scheme. ML-KEM-768~\cite{fips203,kyber} is a NIST Category-3 scheme; its
dominant attack is a primal BKZ reduction whose core operation is lattice
sieving in block size $\beta\approx 637$, with classical sieving cost
exponents $2^{0.292\beta}$ in time and $2^{0.2075\beta}$ in
memory~\cite{bdgj16} and quantum exponent $2^{0.2653\beta}$~\cite{laa15}.
Under five accounting conventions, each of which appears in the
respected literature, the scheme's estimate and its position relative to
its own Category-3 anchor (AES-192) are as follows.

\begin{table}[t]
\centering\footnotesize
\caption{ML-KEM-768 versus its Category-3 anchor (AES-192) under five
literature-anchored accountings ($\beta=637$; bits are $\log_2$ attack
cost). Tags: (S)~specification-published, (D)~derived arithmetic from
cited exponents, (F)~regenerated from cited sources by the toolchain.}
\label{tab:anatomy}
\begin{tabular}{@{}p{4.2cm}p{2.6cm}rrr@{}}
\toprule
\textbf{Cost model} & \textbf{Source} & \textbf{ML-KEM-768} & \textbf{AES-192} & \textbf{Rel.}\\
\midrule
Classical, time-only (Core-SVP) & \cite{adps16,kyber} & $186$\pS & $192$ & $-6$\\
Classical, time$\times$memory & \cite{bdgj16,ecrypt} & $318.18$\pD & $192$ & $+126.18$\\
Quantum, unbounded depth, t-only & \cite{laa15} & $168.99$\pD & $96$ & $+72.99$\\
Quantum, depth-bounded & \cite{grassl16,js19} & $169$--$186$\pF & regime-dep.\pF & regime-dep.\pF\\
Quantum, gate-count & \cite{grassl16,js19} & $3$xx\pF & $1$xx\pF & $+2$xx\pF\\
\bottomrule
\end{tabular}
\end{table}

Three observations follow. \emph{Dispersion:} the anchor-relative security
of one scheme ranges from $-6$ to $+126$ bits across documented
conventions. \emph{Anchor drift:} the anchor itself moves by up to $96$ bits between classical and unbounded-depth quantum accounting, and the
category definition of NIST uses depth-bounded anchors~\cite{nistir8100}
while scheme estimates were predominantly computed under classical
time-only conventions---two different regions of the cost-model space,
compared as if commensurable. \emph{Convention status:} no convention in the table is a strawman;
every one has published defenders. A disagreement about which convention
to use is therefore a dispute about location in a structured parameter
space, not a factual error by either party.

\subsection{From scalars to profiles}

GAUGE (\emph{Generalized Accounting of Uncertainty in Guessing Effort})
makes the cost model a first-class mathematical object and the security
estimate a function of it. An adversary cost model pairs a machine class
with a price vector over canonical adversarial resources; a scheme's
security profile maps each admissible cost model to the $\log_2$ cost of
its cheapest attack under that model; orderings, sensitivities, and risk
measures are then defined over the profile rather than over a scalar.
Figure~\ref{fig:flow} summarizes the framework's structure.

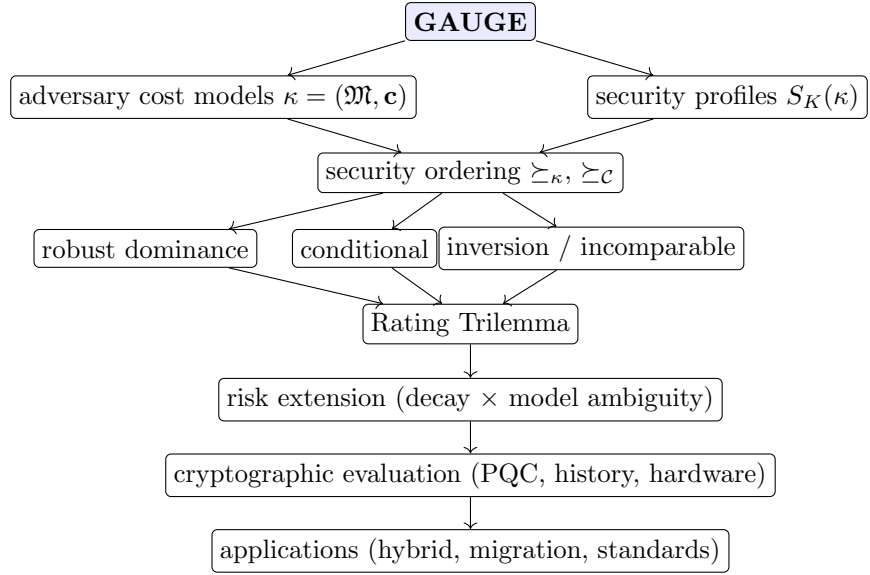
\begin{figure}[t]
\centering
\begin{tikzpicture}[every node/.style={align=center, font=\small},
  box/.style={draw, rounded corners=2pt, minimum height=0.5cm, inner sep=3pt}]
\node[box, fill=blue!8] (root) at (0,0) {\textbf{GAUGE}};
\node[box] (cm) at (-3.4,-1.0) {adversary cost models $\kappa=(\mathfrak M,\mathbf c)$};
\node[box] (sp) at (3.4,-1.0) {security profiles $S_K(\kappa)$};
\node[box] (ord) at (0,-2.0) {security ordering $\succeq_\kappa$, $\succone$};
\node[box] (dom) at (-4.3,-3.0) {robust dominance};
\node[box] (cond) at (-1.4,-3.0) {conditional};
\node[box] (inc) at (1.6,-3.0) {inversion / incomparable};
\node[box] (tri) at (0,-4.0) {Rating Trilemma};
\node[box] (risk) at (0,-5.0) {risk extension (decay $\times$ model ambiguity)};
\node[box] (ev) at (0,-6.0) {cryptographic evaluation (PQC, history, hardware)};
\node[box] (app) at (0,-7.0) {applications (hybrid, migration, standards)};
\draw[->] (root) -- (cm); \draw[->] (root) -- (sp);
\draw[->] (cm) -- (ord); \draw[->] (sp) -- (ord);
\draw[->] (ord) -- (dom); \draw[->] (ord) -- (cond); \draw[->] (ord) -- (inc);
\draw[->] (dom) -- (tri); \draw[->] (cond) -- (tri); \draw[->] (inc) -- (tri);
\draw[->] (tri) -- (risk); \draw[->] (risk) -- (ev); \draw[->] (ev) -- (app);
\end{tikzpicture}
\caption{Conceptual structure of the framework.}
\label{fig:flow}
\end{figure}

\subsection{Contributions}

The contributions of this work are:
\begin{enumerate}[leftmargin=1.5em,itemsep=1pt]
\item We introduce GAUGE, a formal representation of cryptographic
security as a function over heterogeneous adversary cost models
(Sections~\ref{sec:models}--\ref{sec:profiles}), with structural
properties (piecewise linearity, monotonicity, anchor normalization)
established as Propositions~\ref{prop:structure}--\ref{prop:monotone}.
\item We develop security-profile ordering and characterize conditions
under which cryptographic rankings become cost-model dependent
(Theorem~\ref{thm:ordering}); we define the resulting
robust\slash conditional\slash incomparable classification and prove it decidable in
polynomial time by linear programming
(Proposition~\ref{prop:classification}).
\item We formalize the Rating Trilemma and establish the limitations of
scalar, model-independent security ratings
(Theorem~\ref{thm:trilemma}), together with a coherent two-layer risk
extension separating cost-model ambiguity from cryptanalytic decay
(Proposition~\ref{prop:risk}).
\item We evaluate the framework on NIST post-quantum standards and a
curated cryptanalysis chronology (Section~\ref{sec:eval}): the LP
certifies ranking reversals among standardized schemes, a family-level
sensitivity constant is identified, and the instrument is calibrated on
the IonQ \texttt{simulator} (aria-1 noise model; 13 jobs; 2026-09-15).
\end{enumerate}

\subsection{Organization}

Section~\ref{sec:related} reviews related work and positions GAUGE
against existing security notions. Sections~\ref{sec:models}--\ref{sec:profiles}
develop the formal apparatus. Section~\ref{sec:ordering} defines security
ordering and its robust classification. Section~\ref{sec:trilemma} proves
the Rating Trilemma. Section~\ref{sec:risk} develops the risk extension
and the temporal evolution of attack effort.
Section~\ref{sec:eval} reports the cryptographic evaluation.
Section~\ref{sec:applications} derives practical implications for hybrid
cryptography and migration. Sections~\ref{sec:impl}--\ref{sec:conclusion}
describe implementation, limitations, and conclusions. Proofs not given
inline appear in Appendix~\ref{app:proofs}.

\section{Background and Related Work}\label{sec:related}

\subsection{Concrete cryptographic security}

The modern treatment of quantitative security statements begins with
definition-based security~\cite{gm84} and practice-oriented provable
security with exact (concrete) reductions~\cite{br93,br96}. Concrete
security statements condition on a fixed adversary and resource model
(bounded time, bounded queries); the long-running debate on reduction
losses~\cite{km05} concerns the gap between nominal and effective
security under such conditioning. GAUGE takes a different approach: the security estimate is represented as
a function over a family of cost models, so the reduction-loss gap
appears as one explicitly priced uncertainty channel
(U6, Section~\ref{sec:risk}).

\subsection{Security levels and work factors}

Key-size selection and work-factor estimation have a long tradition:
Lenstra--Verheul extrapolated hardware progress to select classical key
sizes~\cite{lv01}; ECRYPT's consensus reports aggregate expert judgment
into point recommendations~\cite{ecrypt}; NIST's post-quantum categories
define equivalence classes relative to AES/SHA resource
costs~\cite{nistir8100,nistir8413}. These approaches output scalars
conditional on unstated or partially stated models. GAUGE makes anchor-relative comparison explicit and adds the ordering
theory needed to classify pairs of schemes.

\subsection{Cost models and resource estimation}

Careful accounting of the full cost of cryptanalytic attacks includes
Wiener's dollar-cost analysis~\cite{wiener04}, Bernstein's treatment of
brute force~\cite{bernstein05}, and classical time--memory
trade-offs~\cite{ss81}. Within lattice cryptanalysis, the sieving
literature~\cite{nv08,mv10,bdgj16,laa15} and the LWE
estimator~\cite{aps15} compute best-attack costs under chosen
conventions. Each such tool provides a point estimate; GAUGE computes the geometry
of the space those estimates inhabit.

\subsection{Post-quantum security estimation}

Standardization documents (FIPS~203--205~\cite{fips203,fips204,fips205};
round reports~\cite{nistir8413}; specifications~\cite{kyber,mceliece,hqc25})
state security under specific conventions (predominantly time-only
classical estimates against depth-bounded quantum anchors). GAUGE makes
the convention dependence of these statements explicit and comparable.

\subsection{Quantum-resource security models}

Quantum cost accounting is an active area: AES/SHA circuit
estimates~\cite{grassl16,amy16}, the RAM-model analysis of
Jaques--Schanck~\cite{js19}, factoring resource
estimates~\cite{ge21}, and quantum speedups for information-set
decoding~\cite{bernstein10}. These works fix a cost model and compute
under it; GAUGE treats the choice among such models as data, and
measures its consequences.

\subsection{Risk and uncertainty in security assessment}

Coherent risk measures~\cite{artzner99,ru00,delbaen02}, decision theory
under ambiguity~\cite{walley91,gs89,kmk05,bewley02}, security
economics~\cite{gl02}, quantum-threat timing~\cite{mosca18,gri}, and
measurement uncertainty frameworks~\cite{stevens46,gum} supply mature
mathematics that has not been connected to the strength of
cryptographic assumptions in this combined form. No prior work applying this combination to concrete cryptographic
security estimates was identified. GAUGE draws on coherent risk
measures for the decay layer, credal sets for the model-ambiguity
layer, and the Type~A/B uncertainty separation for the trust budget.

\subsection{Relation to social-choice impossibility results}\label{sec:socialchoice}

Theorem~\ref{thm:ordering} and the Rating Trilemma
(Theorem~\ref{thm:trilemma}) are impossibility results, and readers
familiar with social choice will correctly notice a family resemblance
to Arrow's theorem and to the Szpilrajn extension theorem's failure mode
under cyclic or crossing preferences. The resemblance is real, and we state precisely where it holds and
where it does not. Szpilrajn guarantees that any strict partial order extends
to a total order; the obstruction we exploit is not partial-order
incompleteness but genuine \emph{sign reversal} of a continuous
real-valued function family on a connected domain, so the relevant
extension failure is closer to the impossibility of a continuous
selection from a family of linear orders that cross---a Condorcet-cycle
phenomenon rather than a missing-comparability phenomenon. Arrow's
theorem rules out a social welfare function satisfying four axioms over
\emph{arbitrary} preference profiles of \emph{discrete} voters; the
Rating Trilemma instead concerns a single decision-relevant object (a
convention for a scalar security level) evaluated over a
\emph{continuum} of cost-model ``voters'' with a specific polyhedral
cone structure, and its proof (Appendix~\ref{app:proofs}) uses the
geometry of that cone directly rather than an Arrovian axiomatic
reduction. The mathematical content we claim as new is not the general
fact that crossing objective functions defeat scalarization---that fact
is old and, in the form we use it, elementary---but rather (i) showing
that the specific cost-model geometry of cryptographic security
estimates is rich enough for such crossings to be certified
constructively (not merely possible in principle) via linear
programming over real, standardized schemes, and (ii) the two-layer risk
decomposition of Section~\ref{sec:risk}, which has no direct analogue in
the classical impossibility literature. Theorem~\ref{thm:ordering} serves to license the LP-certified
classification that follows; the classification, evaluation, and risk
layers carry the paper's technical content.

\subsection{Relation to existing security notions}\label{sec:notions}

Table~\ref{tab:notions} positions GAUGE against the security notions a
cryptographer will know.

\begin{table}[t]
\centering\small
\caption{GAUGE relative to existing security notions.}
\label{tab:notions}
\begin{tabular}{@{}p{3.2cm}p{4.4cm}p{5.6cm}@{}}
\toprule
\textbf{Existing concept} & \textbf{What it measures} & \textbf{GAUGE difference}\\
\midrule
Security parameter & target computational security & GAUGE models its cost-model dependence\\
Concrete security & attack success vs.\ resources (fixed model) & GAUGE represents heterogeneous resource pricing across models\\
Asymptotic security & scaling behaviour & GAUGE supports finite-resource comparison on the same footing\\
Work factor & estimated attack effort & GAUGE treats effort as model-dependent (a profile)\\
Security level / category & scalar classification & GAUGE represents the underlying profile and its crossings\\
Risk assessment & decision uncertainty & GAUGE separates cost-model ambiguity (systematic) from cryptanalytic decay (stochastic)\\
PQC security categories & standardized classification & GAUGE analyzes sensitivity to adversary-cost assumptions and exhibits inversions\\
\bottomrule
\end{tabular}
\end{table}

\section{Adversary Cost Models}\label{sec:models}

\begin{definition}[Resource ledger]\label{def:ledger}
A \emph{resource ledger} is a finite ordered set
 $\TT=(R_1,\dots,R_n)$ of canonical adversarial resources. The default
ledger is $\TT_6=(T,M,Q,D,W,N)$: logical time $T$, classical memory $M$,
logical qubits $Q$, quantum circuit depth $D$, quantum width $W$, and
parallel instances $N$. Coherent extensions (energy, leakage traces,
target-count amortization) are admissible; their role in
Section~\ref{sec:trilemma} is not incidental but central.
\end{definition}

\begin{definition}[Adversary cost model]\label{def:model}
An \emph{adversary cost model} is a pair
 $\kappa=(\mathfrak M,\mathbf c)$, where $\mathfrak M$ is a machine class
(classical sequential or parallel; quantum with unbounded depth; quantum
with depth bound $D_{\max}$; quantum with quantum-accessible random
memory~\cite{js19}) and
 $\mathbf c=(c_1,\dots,c_n)\in\mathbb{R}^n_{+}$ is a \emph{price vector}
assigning to each resource a price in units of the time resource (the
\emph{num\'eraire convention} $c_T=1$). The \emph{price cone} is
\[
\Cone_\TT \;=\; \{\mathbf c\in\mathbb{R}^n_{+}\;:\;c_T=1\}.
\]
Time-only pricing corresponds to $\mathbf c=(1,0,\dots,0)$;
time$\times$memory pricing to $c_T=c_M=1$; gate-count pricing
(depth$\times$width) to $c_D=c_W=1$ with the time coordinate suppressed.
\end{definition}

The num\'eraire convention is not cosmetic: it is precisely what makes
zero-memory anchors (e.g.\ AES key search, which uses negligible memory)
price-invariant across time-only and memory-aware conventions, which is
required for anchor-relative comparison to be well defined
(Proposition~\ref{prop:monotone}).

\begin{definition}[Attack record; attack set]\label{def:attack}
An \emph{attack record} for a scheme $K$ is a pair $(A,x_A)$: an
algorithm $A$ with a stated success probability, together with a
log-resource vector $x_A\in\mathbb{R}^n$ (i.e., attack $A$ consumes
 $\bits{x_{A,i}}$ units of resource $R_i$), carrying provenance (paper,
year). The \emph{attack set} $\mathcal{A}_K$ is the set of all such
records at evaluation time; it grows as cryptanalysis progresses. Attack
 $A$ is \emph{feasible} in $\kappa$ if $A$ is admissible in the machine
class $\mathfrak M$ of $\kappa$.
\end{definition}

\begin{definition}[Literature closure]\label{def:closure}
The \emph{admissible set} $\KK$ is the set of cost models each anchored
to a position actually taken in peer-reviewed publications, standards
documents, or estimator tools, with the citation attached to every
element. $\KK$ is an anti-strawman construction: no model in $\KK$ lacks
a published defender.
\end{definition}

\begin{table}[t]
\centering\small
\caption{Named accounting conventions as elements of the price cone,
with their defended positions.}
\label{tab:conventions}
\begin{tabular}{@{}llll@{}}
\toprule
\textbf{Convention} & \textbf{Pricing} & \textbf{Defended in} & \textbf{Lattice attack cost}\\
\midrule
Core-SVP & time-only & \cite{adps16,kyber} & $2^{0.292\beta}$\pS\\
TM-sieving & time$\times$memory & \cite{bdgj16,ecrypt} & $2^{0.4995\beta}$\pD\\
Quantum sieve & time-only, quantum & \cite{laa15} & $2^{0.2653\beta}$\pD\\
Depth-bounded & $D\le D_{\max}$ & \cite{nistir8100,grassl16,js19} & regime-dependent\pF\\
Gate-count & depth$\times$width & \cite{grassl16} & $\gtrsim$ TM\pF\\
RAM-model & quantum, RAM costs & \cite{js19} & revised exponents\pF\\
\bottomrule
\end{tabular}
\end{table}

\section{GAUGE Security Profiles}\label{sec:profiles}

\begin{definition}[Price functional]\label{def:price}
For a cost model $\kappa=(\mathfrak M,\mathbf c)$ and attack $A$ feasible
in $\kappa$, the \emph{price} of $A$ is
\[
P_\kappa(A)\;=\;\mathbf c\cdot x_A\;=\;\sum_{i=1}^n c_i\,x_{A,i}
\qquad(\text{bits}).
\]
\end{definition}

\begin{definition}[Security profile]\label{def:profile}
The \emph{security profile} of scheme $K$ is
\[
S_K(\kappa)\;=\;\min\{\mathbf c\cdot x_A\;:\;(A,x_A)\in\mathcal{A}_K,\
A\ \text{feasible in}\ \kappa\}.
\]
For an \emph{anchor problem} $G$ (e.g., AES-128 exhaustive key search)
with reference attack cost $S_G(\kappa)$, the \emph{anchor-relative
profile} is $\Stil{K}(\kappa)=S_K(\kappa)-S_G(\kappa)$.
\end{definition}

The scalar in a specification sheet is $S_K(\kappa_0)$ for an unstated
 $\kappa_0$; the profile makes the dependence explicit
(Figure~\ref{fig:profiles}).

\begin{figure}[t]
\centering
\begin{tikzpicture}[scale=1.0]
\draw[->] (0,0) -- (7.0,0) node[below] {$c_M$ (memory price)};
\draw[->] (0,0) -- (0,4.3);
\node[left] at (0,0.55) {$192$};
\node[left] at (0,0.40) {$186$};
\node[left] at (0,3.70) {$318$};
\draw[dashed,thick] (0,0.55) -- (6.6,0.55);
\node[above] at (4.6,0.55) {\small AES-192 (anchor: memory-free)};
\draw[thick] (0,0.40) -- (6.6,3.70);
\node[below left] at (5.9,3.10) {\small ML-KEM-768};
\fill (0.30,0.55) circle (1.7pt);
\draw[dotted] (0.30,0) -- (0.30,0.55);
\node[below] at (0.30,-0.06) {\scriptsize $c_M^{*}=0.045$};
\node[below] at (0.05,-0.06) {\scriptsize $0$ (Core-SVP)};
\node[below] at (6.35,-0.06) {\scriptsize $1$ (TM)};
\end{tikzpicture}
\caption{Security profiles (absolute bits) on the classical memory-price
slice: $S_{\text{AES-192}}$ is constant $192$; $S_{\text{ML-KEM-768}}
=186+132.18\,c_M$. The profiles cross at $c_M^{*}=0.045$; pricing memory
at $4.5\%$ of the time unit reverses the comparison
(Theorem~\ref{thm:ordering}, Experiment~2).}
\label{fig:profiles}
\end{figure}
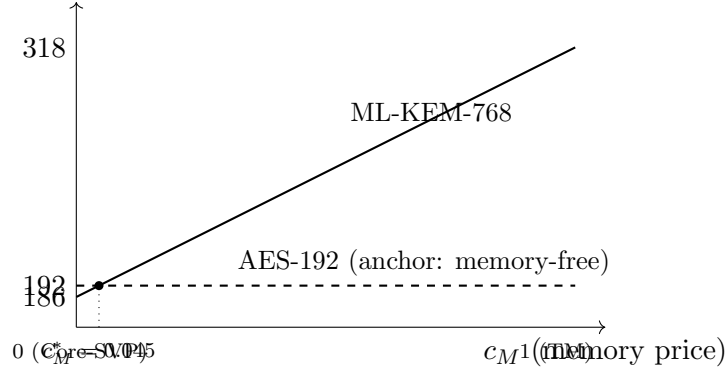

\begin{proposition}[Structure of security profiles]\label{prop:structure}
Fix an attack set with log-resource vectors $\{x_i\}_{i\le m}$ and the
price cone $\Cone_\TT$. Then:
\emph{(i)} $S_K(\mathbf c)=\min_i\,\mathbf c\cdot x_i$;
\emph{(ii)} $S_K$ is concave, positively homogeneous, and piecewise
linear on $\Cone_\TT$; its linearity regions are precisely the normal
cones of the \emph{attack polytope}
 $P_K=\mathrm{conv}\{x_i\}$~\cite{klst71};
\emph{(iii)} for any polyhedral region $\Cone\subseteq\Cone_\TT$, both
 $\max_{\Cone}S_K$ and $\min_{\Cone}S_K$ are attained at vertices of
 $\Cone$ and are computable by vertex enumeration in time
 $O(|V(\Cone)|\cdot m)$;
\emph{(iv)} $S_K$ is a tropical polynomial with Newton polytope $P_K$.
\end{proposition}

\begin{proof}
(i) is the definition. (ii) A pointwise minimum of linear functionals is
concave and positively homogeneous; the linearity regions of
 $\min_i\,\mathbf c\cdot x_i$ are the normal cones of
 $\mathrm{conv}\{x_i\}$. (iii) Concavity gives, for
 $\mathbf c=\sum_k\lambda_k\mathbf v_k$ a convex combination of vertices,
 $S_K(\mathbf c)\ge\sum_k\lambda_k S_K(\mathbf v_k)\ge\min_k
S_K(\mathbf v_k)$, so the minimum over $\Cone$ equals the vertex
minimum; the maximum follows symmetrically since a linear functional
attains its maximum over a polytope at a vertex and
 $\max_{\Cone}\min_i\mathbf c\cdot x_i$ is attained where some
 $\mathbf c\cdot x_i$ is maximized. (iv) Immediate from (ii).
\end{proof}

\begin{proposition}[Monotonicity and normalization]\label{prop:monotone}
\emph{(i) Price monotonicity:} if $\mathbf c\le\mathbf c'$ componentwise then $S_K(\mathbf c)\le S_K(\mathbf c')$.
\emph{(ii) Knowledge monotonicity:} if
 $\mathcal{A}_K\subseteq\mathcal{A}'_K$ then
 $S_K\ge S'_K$: adding attack knowledge only decreases the known-attack
profile, which is therefore an upper bound on true security that
tightens with cryptanalytic progress.
\emph{(iii) Anchor-normalization invariance:} if the effort statistic is
rescaled by a constant factor (equivalently, $2^{S}$ is replaced by
 $\lambda\,2^{S}$ for a monotone $1$-homogeneous effort functional), the
anchor-relative profile $\Stil{K}$ is unchanged. In particular,
 $\Stil{K}$ is the invariant observable; absolute bits are
convention-dependent up to the anchor.
\end{proposition}

\begin{proof}
(i) Each $\mathbf c\cdot x_i$ is nondecreasing in $\mathbf c$; minima
preserve this. (ii) Minimizing over a larger set can only decrease the
value. (iii) $\Stil{K}=\log_2(2^{S_K}/2^{S_G})$; positive $1$-homogeneity
of the effort functional cancels in the ratio.
\end{proof}

\begin{remark}[Semantics of the unit]\label{rem:unit}
For generic search problems the operational semantics of ``$n$ bits of
security'' is the $\log_2$ of expected optimal guessing
effort~\cite{massey94,arikan96}, which is the adversary's actual
expenditure; Proposition~\ref{prop:monotone}(iii) then says that only
anchor-relative differences of this quantity are invariant. A thermodynamic
floor exists as well---each irreversible bit operation dissipates at least
 $kT\ln 2$ joules~\cite{landauer61,bennett89}---but we use it only as a
unit anchor, not as an estimate.
\end{remark}

\begin{definition}[Fragility; margins]\label{def:fragility}
The \emph{fragility index} of $K$ over $\KK$ is
 $\Delta_K=\max_{\kappa,\kappa'\in\KK}|\Stil{K}(\kappa)-\Stil{K}(\kappa')|$,
and the \emph{relative fragility} is
 $\Delta_K/\sigma^{\mathrm{nom}}_K$, where $\sigma^{\mathrm{nom}}_K$ is
the claimed level under the scheme's declared (time-only) convention.
The \emph{attack margin} is the gap between the claimed level and the
best known attack under the declared model; the \emph{evidential margin}
is the gap to the best provable lower bound---close to the full claim
for essentially all deployed primitives, including AES and ML-KEM.
\end{definition}

\section{Security Ordering}\label{sec:ordering}

\subsection{Pointwise and robust orderings}

\begin{definition}[Pointwise security ordering]\label{def:pointwise}
For schemes $K_1,K_2$ and cost model $\kappa\in\KK$:
\[
K_1\succeq_\kappa K_2 \quad\Longleftrightarrow\quad
\Stil{K_1}(\kappa)\ge \Stil{K_2}(\kappa).
\]
\end{definition}

\begin{definition}[Robust security ordering; classification]\label{def:robust}
For a region $\Cone\subseteq\Cone_\TT$:
 $K_1\succone K_2$ (\emph{$K_1$ dominates $K_2$ on $\Cone$}) iff
 $\Stil{K_1}(\kappa)\ge\Stil{K_2}(\kappa)$ for all $\kappa\in\Cone$. For
generic profiles, each pair falls into exactly one of:
\begin{itemize}[leftmargin=1.5em,itemsep=0pt]
\item \textbf{Robust dominance}: $K_1\succone K_2$ with strict
inequality somewhere on $\Cone$ ($K_1$ wins throughout);
\item \textbf{Conditional dominance}: $K_1\succone K_2$ with strict
inequality only on a proper subregion (superiority holds only for models
in that subregion; elsewhere the schemes are measurement-equivalent);
\item \textbf{Incomparable}: there exist $\kappa,\kappa'\in\Cone$ with
 $\Stil{K_1}(\kappa)>\Stil{K_2}(\kappa)$ and
 $\Stil{K_1}(\kappa')<\Stil{K_2}(\kappa')$ (the ordering reverses).
\end{itemize}
The degenerate case $\Stil{K_1}\equiv\Stil{K_2}$ on $\Cone$ is
\emph{measurement equivalence}: the schemes are indistinguishable by the
instrument (realized in Section~\ref{sec:eval} by AES-128 and
SLH-DSA-128s).
\end{definition}

A category claim (``$K$ is in Category $c$'') is, in these terms, the
statement $\Stil{K}(\kappa)\ge 0$ against the category-$c$ anchor
\emph{over an implicitly quantified region} $\Cone$ that is almost never
stated. Definition~\ref{def:robust} forces the region into the open.

\subsection{Cost-model-dependent orderings}

\begin{theorem}[Cost-model-dependent security ordering]\label{thm:ordering}
Let $K_1,K_2$ be schemes with profiles over an admissible set $\KK$.
Call a total preorder $\succeq$ on schemes \emph{sound} if
 $K\succeq K'$ implies $\Stil{K}(\kappa)\ge\Stil{K'}(\kappa)$ for all
 $\kappa\in\KK$ (every ranking it asserts is supported at every
admissible model). If there exist $\kappa_a,\kappa_b\in\KK$ with
\[
\Stil{K_1}(\kappa_a)>\Stil{K_2}(\kappa_a)
\qquad\text{and}\qquad
\Stil{K_1}(\kappa_b)<\Stil{K_2}(\kappa_b),
\]
then no sound total ordering of $K_1$ and $K_2$ exists: any total
ranking of the pair is unsupported at some admissible cost model.
\end{theorem}

\begin{proof}
By totality, $K_1\succeq K_2$ or $K_2\succeq K_1$; assume the former
without loss of generality. Soundness yields
 $\Stil{K_1}(\kappa)\ge\Stil{K_2}(\kappa)$ for all $\kappa\in\KK$,
contradicting the strict reversal at $\kappa_b$. The symmetric case
contradicts the reversal at $\kappa_a$.
\end{proof}

Theorem~\ref{thm:ordering} is not merely a possibility result:
Section~\ref{sec:eval} exhibits such pairs among NIST-standardized
schemes, with linear-programming certificates. It also motivates the
classification of Definition~\ref{def:robust}: since sound total
orderings fail exactly at crossing pairs, the well-defined questions
become (i) does a crossing exist in a given region, and (ii) if not,
which scheme dominates. Both are decidable efficiently; we make no claim
that this decidability is itself a deep algorithmic result; the
polynomial-time bound below is an immediate consequence of casting
crossing-detection as linear-programming feasibility, and we state it
because the reduction is the useful part---it is what turns "these
schemes might disagree" into a certificate a standards body can check,
not because bounding an LP's complexity is difficult.

\subsection{Complexity of dominance and classification}

\begin{proposition}[Decidability and complexity of the
classification]\label{prop:classification}
Let $\Cone\subseteq\Cone_\TT$ be polyhedral and let
 $\mathcal{A}_{K_1},\mathcal{A}_{K_2}$ be finite attack sets with
 $m_1,m_2$ records. Then:
\emph{(i)} whether there exists $\kappa\in\Cone$ with
 $\Stil{K_1}(\kappa)<\Stil{K_2}(\kappa)$ is decidable by solving $m_1$ linear programs
\[
\min_{\mathbf c,t}\ t\quad\text{s.t.}\quad
(x^{K_1}_i-x^{K_2}_j)\cdot\mathbf c\le t\ \ \forall j\in\mathcal{A}_{K_2},
\qquad \mathbf c\in\Cone,
\]
and testing whether some optimum satisfies $t^*<0$ (with the
witness $\mathbf c^*$ the certificate);
\emph{(ii)} the classification of
Definition~\ref{def:robust} is decidable by running the test of (i) in
both directions, using $O(m_1+m_2)$ linear programs of polynomial size;
hence robust-dominance testing over the price cone reduces to checking
finitely many polyhedral extrema and is solvable in polynomial time.
\end{proposition}

\begin{proof}[Proof sketch]
 $\Stil{K_1}(\mathbf c)<\Stil{K_2}(\mathbf c)$ iff some attack $i$ of
 $K_1$ satisfies $\mathbf c\cdot x^{K_1}_i<\min_j \mathbf c\cdot
x^{K_2}_j$, i.e., $\max_j (x^{K_1}_i-x^{K_2}_j)\cdot\mathbf c<0$;
minimizing the left side over $\Cone$ is exactly the displayed
epigraph LP. The classification combines the two directional tests;
non-strict boundaries are handled by $\varepsilon$-slack in the standard
way. Full proof in Appendix~\ref{app:proofs}.
\end{proof}

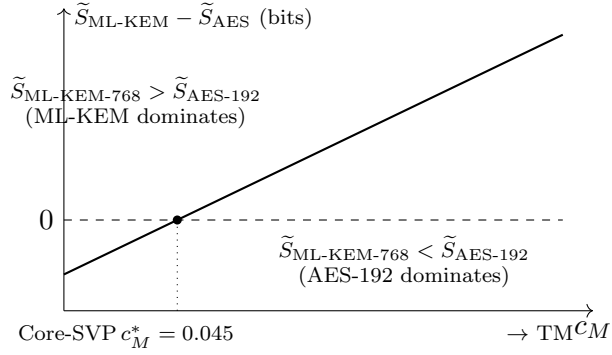
\begin{figure}[t]
\centering
\begin{tikzpicture}
\draw[->] (0,0) -- (7.0,0) node[below] {$c_M$};
\draw[->] (0,0) -- (0,4.0);
\node[left] at (0,1.2) {$0$};
\draw[dashed] (0,1.2) -- (6.6,1.2);
\draw[thick] (0,0.48) -- (6.6,3.65);
\fill (1.50,1.2) circle (1.7pt);
\draw[dotted] (1.50,0) -- (1.50,1.2);
\node[below] at (1.50,-0.06) {\scriptsize $c_M^{*}=0.045$};
\node[below] at (0.05,-0.06) {\scriptsize Core-SVP};
\node[below] at (6.3,-0.06) {\scriptsize $\to$ TM};
\node[align=center, font=\scriptsize] at (0.95,2.75)
{$\Stil{\text{ML-KEM-768}}>\Stil{\text{AES-192}}$\\ (ML-KEM dominates)};
\node[align=center, font=\scriptsize] at (4.5,0.65)
{$\Stil{\text{ML-KEM-768}}<\Stil{\text{AES-192}}$\\ (AES-192 dominates)};
\node[right] at (0,3.9) {\scriptsize $\Stil{\text{ML-KEM}}-\Stil{\text{AES}}$ (bits)};
\end{tikzpicture}
\caption{The price-cone region map for the pair (ML-KEM-768,
AES-192) on the classical memory-price slice. The anchored difference
is $-6+132.18\,c_M$; the ordering reverses across $c_M^{*}=0.045$. The
pair is incomparable on any region containing both conventions
(Proposition~\ref{prop:classification}, Theorem~\ref{thm:ordering}).}
\label{fig:regions}
\end{figure}

\section{The Rating Trilemma}\label{sec:trilemma}

A \emph{rating} assigns to each scheme a real number determined by its
profile, inducing an order $K\succeq_R K'$ iff $R(K)\ge R(K')$. Three
properties are natural:

\begin{axiom}[P1: Faithfulness]\label{ax:faith}
If $\Stil{K}(\kappa)\ge\Stil{K'}(\kappa)$ for all $\kappa\in\KK$, with
strict inequality somewhere, then $R(K)\ge R(K')$.
\end{axiom}
\begin{axiom}[P2: Totality]\label{ax:total}
 $\succeq_R$ is a total preorder: every pair of schemes receives an
ordering.
\end{axiom}
\begin{axiom}[P3: Model independence]\label{ax:obj}
 $R$ is invariant under the \emph{convention group} $\mathscr G$:
(i) relabelings and rescalings of resource prices within the ledger
(projectively), (ii) permutations of ledger resources (including the
choice of num\'eraire), and (iii) coherent extensions
 $\TT\to\TT'$ of the ledger (adjoining further physical resources).
\end{axiom}

P2 is automatic for scalar ratings---which is precisely the problem: the
underlying order (Section~\ref{sec:ordering}) is \emph{not} total once
crossings exist, and Theorem~\ref{thm:trilemma} identifies the price
totality exacts.

\begin{lemma}[Mixture representation]\label{lem:mixture}
Let $R$ be linear in the profile and normalized. Then $R$ satisfies
P1 iff there exists a probability vector $\mu$ on $\KK$ with
\[
R(K)\;=\;\sum_{\kappa\in\KK}\mu_\kappa\,\Stil{K}(\kappa)
\;=\;\Ex_{\kappa\sim\mu}\!\big[\Stil{K}(\kappa)\big].
\]
\end{lemma}

\begin{proof}
Linearity on the finite-dimensional profile space gives a representing
vector $\mu$. P1 is monotonicity along every nonnegative direction;
testing on pairs differing only at a single $\kappa$ forces
 $\mu_\kappa\ge0$, and normalization gives $\sum_\kappa\mu_\kappa=1$.
\end{proof}

\begin{lemma}[Symmetry rigidity]\label{lem:symmetry}
Let $\mu$ be a Borel probability measure on the projective price space
invariant under the rescaling flow of every resource $r$ ($\omega_r\mapsto e^t\omega_r$, renormalized). Then $\mu$ is supported on
the pure single-resource price vectors $\{e_r\}$. If $\mu$ is
additionally invariant under ledger permutations, it is uniform over
them.
\end{lemma}

\begin{proof}[Proof sketch]
By Poincar\'e recurrence, $\mu$-a.e.\ point is recurrent under the
time-one map of each flow; any point with two strictly positive
coordinates flows forward to a pure vertex and backward to a coordinate
face, hence is not recurrent. Intersecting over resources leaves exactly
the pure vectors, and permutation invariance on this transitive finite
orbit forces uniformity. Full proof in Appendix~\ref{app:proofs}.
\end{proof}

\begin{lemma}[Ledger dependence]\label{lem:ledger}
\emph{(i)} The uniform-over-pure rating of Lemma~\ref{lem:symmetry}
changes value under coherent ledger extensions whenever profiles are
non-constant on $\KK$, and such non-constancy holds for standardized
schemes (Section~\ref{sec:eval}). \emph{(ii)} Every
extension-invariant linear rating is supported on a fixed a priori set
of models, a commitment no measurement determines.
\end{lemma}

\begin{theorem}[Rating Trilemma]\label{thm:trilemma}
Within the class of ratings that are monotone aggregators of the
profile---in particular all linear ratings and the minimax rating---no
rating satisfies P1, P2, and P3 simultaneously once the admissible set
contains schemes with crossing profiles. Concretely:
\emph{(a)} by Lemma~\ref{lem:mixture}, every total faithful linear
rating is a probability mixture over adversary cost models;
\emph{(b)} by Lemmas~\ref{lem:symmetry}--\ref{lem:ledger}, P3's
rescaling-and-permutation invariance forces the uniform-over-pure
mixture, which is not extension-invariant, while extension-invariant
mixtures require an a priori support commitment;
\emph{(c)} the minimax rating $R=\min_\kappa\Stil{K}(\kappa)$ is
permutation-invariant but strictly decreases under ledger extensions
that add a model where the profile is smaller, violating P3(iii).
Consequently every faithful total rating encodes a choice of
distribution over---or selection among---adversary cost models; that
choice is a modeling decision, not a measurement.
\end{theorem}

\begin{proof}
Assemble (a)--(c): a P1--P2 linear rating is a mixture
(Lemma~\ref{lem:mixture}); P3(i)--(ii) forces uniform-over-pure
(Lemma~\ref{lem:symmetry}); P3(iii) then fails by
Lemma~\ref{lem:ledger}(i), while evading via
Lemma~\ref{lem:ledger}(ii) contradicts the invariance requirement
itself. The minimax case is (c). Full assembly in
Appendix~\ref{app:proofs}.
\end{proof}

\begin{corollary}[Rating-as-distribution]\label{cor:dist}
Every faithful, total, linear rating with normalization admits the
interpretation of an expectation under a prior $\mu$ over $\KK$; two
raters who disagree do so only through their $\mu$.
\end{corollary}

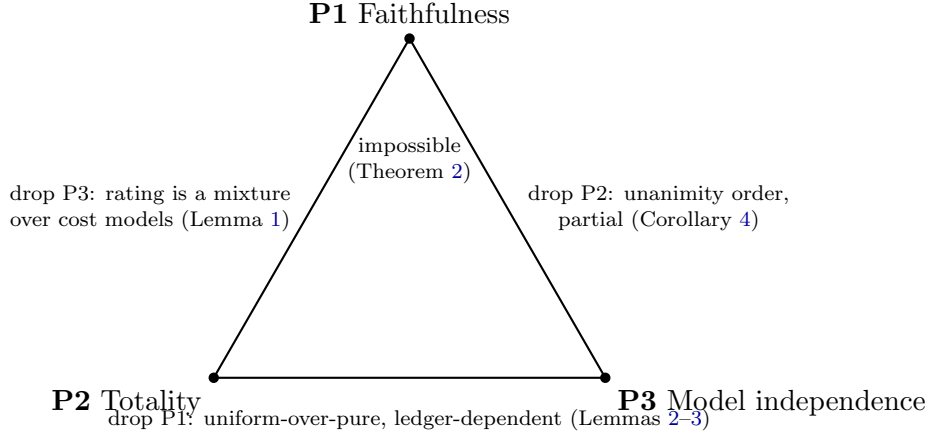
\begin{figure}[t]
\centering
\begin{tikzpicture}[scale=1.15]
\coordinate (A) at (90:2.6);
\coordinate (B) at (210:2.6);
\coordinate (C) at (330:2.6);
\draw[thick] (A) -- (B) -- (C) -- cycle;
\fill (A) circle (1.7pt) node[above=2pt] {\textbf{P1} Faithfulness};
\fill (B) circle (1.7pt) node[below left=1pt] {\textbf{P2} Totality};
\fill (C) circle (1.7pt) node[below right=1pt] {\textbf{P3} Model independence};
\node[align=center, font=\scriptsize] at (90:1.2) {impossible\\ (Theorem~\ref{thm:trilemma})};
\node[align=center, font=\scriptsize] at ($(A)!0.5!(B)$) [left=4pt]
{drop P3: rating is a mixture\\ over cost models (Lemma~\ref{lem:mixture})};
\node[align=center, font=\scriptsize] at ($(B)!0.5!(C)$) [below=8pt]
{drop P1: uniform-over-pure, ledger-dependent (Lemmas~\ref{lem:symmetry}--\ref{lem:ledger})};
\node[align=center, font=\scriptsize] at ($(A)!0.5!(C)$) [right=4pt]
{drop P2: unanimity order,\\ partial (Corollary~\ref{cor:bewley})};
\end{tikzpicture}
\caption{The Rating Trilemma. Each edge is a viable research program
obtained by relaxing one axiom; the center is impossible in the presence
of crossing profiles, which Section~\ref{sec:eval} exhibits among
standardized schemes.}
\label{fig:trilemma}
\end{figure}

\begin{remark}[Social choice]\label{rem:arrow}
With schemes as alternatives and cost models as voters, crossing
profiles generate Condorcet-type cycles among total ratings; the
trilemma is an Arrow-flavored phenomenon~\cite{arrow51} in measurement
form. We note the correspondence without claiming a formal equivalence.
\end{remark}

\section{Security Uncertainty and Temporal Evolution}\label{sec:risk}

\subsection{Uncertainty channels}

GAUGE separates six channels: \textbf{U1} cost-model choice (systematic);
\textbf{U2a} cryptanalytic \emph{drift} (gradual improvement within the
known attack set, e.g., sieve-exponent progress); \textbf{U2b}
cryptanalytic \emph{jumps} (new attack families, e.g.,
\cite{castryck23,beullens22}); \textbf{U3} hardware drift (relative
resource prices); \textbf{U4} quantum-arrival timing; \textbf{U5}
implementation risk (e.g., \cite{kyberslash}); \textbf{U6} reduction-loss
uncertainty~\cite{br96,km05}. U2--U4 are stochastic (Type~A); U1 and U6
are systematic (Type~B); U5 straddles both~\cite{gum}.

The price-cone geometry of
Sections~\ref{sec:models}--\ref{sec:trilemma} concerns channel
\textbf{U1} alone: it takes the set of known attacks as given and
studies what happens when the \emph{prices} of their resources change.
Every historical failure in the CryChron chronology---SHA-1, SIDH,
Matsumoto--Imai, Rainbow---was not a repricing of a known attack but the
arrival of a previously unknown one (channel \textbf{U2b}). Channel U2b
is handled separately through the survival model of
Section~\ref{sec:dynamics}. The geometric layer addresses the distinct
question of whether two schemes can be ranked inconsistently depending
on an accounting convention, independent of any new cryptanalysis.

\subsection{A two-layer risk measure}

The two mathematically distinct uncertainties must not be conflated:
\emph{given} a cost model, decay is stochastic; \emph{which} cost model
is appropriate is ambiguity~\cite{walley91,gs89,kmk05}. For scheme $K$ with target level $s^*$, let the loss process be
 $L_K(t)=[\,s^*-\Stil{K}(t)\,]^+$ under the decay processes of its
assumption families.

\begin{definition}[Credal set; trust budget; GAUGE-risk]\label{def:risk}
Let $\CC_\rho=\{\mu:D_{\mathrm{KL}}(\mu\|\mu_0)\le\rho\}$ be a credal
set of priors over $\KK$ around a reference prior $\mu_0$ (\emph{trust budget} $\rho\ge0$, in nats). For horizon $\tau$ and
confidence $\delta$:
\[
\rho(K;\tau,\delta,\rho)\;=\;\sup_{\mu\in\CC_\rho}\
\CVaR^{\mu}_{\delta}\big(L_K(\tau)\big).
\]
\end{definition}

\begin{proposition}[Two-layer risk and its properties]\label{prop:risk}
\emph{(i) Coherence:} $\rho(K)$ is a coherent risk measure
(monotone, positively homogeneous, translation-invariant,
subadditive)~\cite{artzner99,ru00}.
\emph{(ii) Decomposition:} the inner $\CVaR$ aggregates stochastic decay
given a cost model; the outer supremum aggregates systematic ambiguity
across cost models; $\rho=0$ recovers the single-model risk, and the
trust budget prices the Type-B uncertainty of the apparatus.
\emph{(iii) Tractability:} on finite scenario spaces, $\rho$ is
computable by convex programming; with KL balls, via the dual
\[
\rho(K)=\min_{t}\Big\{t+\tfrac1\delta\inf_{\beta>0}\big[\beta\rho+
\beta\log\Ex_{\mu_0}e^{(L-t)^+/\beta}\big]\Big\},
\]
by Sion's minimax theorem~\cite{sion58} and the variational dual of the
KL ball~\cite{ek18}.
\end{proposition}

\begin{proof}[Proof sketch]
Each $\CVaR^\mu_\delta$ is coherent; a pointwise supremum of coherent
measures preserves monotonicity, homogeneity, and translation
invariance, and subadditivity follows from subadditivity of each
 $\CVaR$ plus $\sup_\mu[f+g]\le\sup_\mu f+\sup_\mu g$. The dual is the
Rockafellar--Uryasev representation of $\CVaR$ composed with the KL-ball
variational form. Full proof in Appendix~\ref{app:proofs}.
\end{proof}

The corresponding claim format for specification sheets is:
\emph{``$K$ retains $\ge s$ anchor-relative bits at horizon $\tau$,
confidence $\delta$, within trust budget $\rho$.''} A worked illustration
with declared illustrative inputs appears in Appendix~\ref{app:extra}.

\subsection{Temporal evolution of attack effort}\label{sec:dynamics}

\paragraph{Cryptanalysis chronology.}
We curate a versioned dataset (\emph{CryChron}; full contents and
methodology in Appendix~\ref{app:data}) of $22$ assumption generations
across nine strata, with $10$ observed break (``jump'') events and $12$ right-censored generations. Observed ages at break range from $6$ to
 $22$ years. The pooled Kaplan--Meier estimator~\cite{kaplan58} crosses
 $50\%$ survival at $t=22$ years
($\hat S(22)=0.446$; Greenwood variance $0.0758$; bootstrap $95\%$ CI
 $\approx[0.26,0.77]$); the stratum-level curves for the event-free
strata (LWE-type lattices, codes, factoring, discrete logarithms) remain
at $1$ throughout the window, so their medians are not estimable---a
statement about the state of evidence. Because cryptanalytic events
cluster by technique (the 2004 Wang-cluster breaks; the
GGH/NTRUSign transcript-attack lineage), the survival model carries a
shared frailty per technique lineage.

\paragraph{Construction and scope of the drift statistic.}
The historical lattice drift---$78.35$ bits gained on the reference
block size $\beta=637$ between 2008 and 2016, i.e., $9.79$ bits per
year (Figure~\ref{fig:evolution})---is a \emph{descriptive} statistic of
a single attack family over a single window. Its construction and limits
are as follows. \emph{Dataset:} peer-reviewed improvements to
lattice-sieving exponents only ($0.415\to0.292$ \cite{nv08,mv10,bdgj16}). \emph{Normalization:}
all magnitudes are evaluated on one reference parameter class
($\beta=637$, the ML-KEM-768 block size); the figure is not a rate per
cryptographic family and is not comparable across families (a $2$-bit
hash-function improvement has a different denominator and mechanism).
\emph{Independence:} successive sieve improvements build on their
predecessors, so the observations are not independent---which is why we
model them as a renewal process rather than a regression, and why no
confidence interval is attached to the raw rate (with four events it
would be meaningless). \emph{Descriptive, not predictive:} we do not
extrapolate the $9.79$ figure beyond the sieving family and the
2008--2016 window; predictive use is confined to the explicit renewal
model below, whose prior is stated and whose output is a distribution.

\paragraph{Renewal posterior.}
Modeling classical sieve-exponent improvements as a Poisson process
with a Gamma prior calibrated to the 2008--2016 window
($\mathrm{Gamma}(4,8)$; prior choice illustrative, arithmetic exact;
sensitivity: across a $5\times4$ grid of $\mathrm{Gamma}(\alpha,\beta)$
priors with $\alpha\in\{2,3,4,5,6\}$, $\beta\in\{4,8,12,16\}$ the
first-event drift ranges from $9.4$ to $19.8$ bits, stable within
$1.4\times$ of the baseline), the
nine subsequent quiet years update the posterior to
 $\mathrm{Gamma}(4,17)$: mean rate $0.235$ events/year;
 $\Prob[\text{no event in 5y}]=(17/22)^4\approx0.357$; mean event
magnitude $26.12$ bits on $\beta=637$; and the \emph{first-event}
drift---the expected magnitude of the first improvement in the
window---is $16.8$ bits over five years, while the uncapped expectation
 $E[N(5)]\cdot\bar m=30.7$ bits counts subsequent events at full
magnitude and is an upper estimate, since post-plateau steps have
historically been smaller.

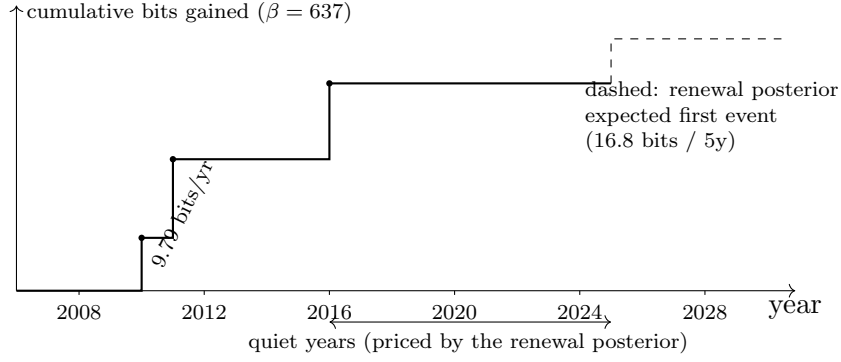
\begin{figure}[t]
\centering
\begin{tikzpicture}[scale=0.92]
\draw[->] (0,0) -- (11.2,0) node[below] {year};
\draw[->] (0,0) -- (0,4.1);
\foreach \yr/\xt in {2008/0.9, 2012/2.7, 2016/4.5, 2020/6.3, 2024/8.1, 2028/9.9}
  \draw (\xt,0) -- (\xt,-0.06) node[below]{\scriptsize \yr};
\draw[thick] (0,0) -- (1.8,0) -- (1.8,0.76) -- (2.25,0.76)
  -- (2.25,1.89) -- (4.5,1.89) -- (4.5,2.98) -- (8.55,2.98);
\draw[dashed] (8.55,2.98) -- (8.55,3.62) -- (11.0,3.62);
\fill (1.8,0.76) circle (1.3pt); \fill (2.25,1.89) circle (1.3pt);
\fill (4.5,2.98) circle (1.3pt);
\node[right] at (0,4.0) {\scriptsize cumulative bits gained ($\beta=637$)};
\node[rotate=64] at (2.4,1.1) {\scriptsize $9.79$ bits/yr};
\draw[<->] (4.5,-0.45) -- (8.55,-0.45);
\node[below] at (6.5,-0.45) {\scriptsize quiet years (priced by the renewal posterior)};
\node[align=left, font=\scriptsize] at (10.0,2.5)
{dashed: renewal posterior\\ expected first event\\ ($16.8$ bits / 5y)};
\end{tikzpicture}
\caption{Historical evolution of lattice-sieving attack cost on the
reference parameter class $\beta=637$: cumulative bits gained by
exponent improvements, 2008--2016, followed by nine quiet years.}
\label{fig:evolution}
\end{figure}

\begin{figure}[t]
\centering
\includegraphics[width=0.75\textwidth]{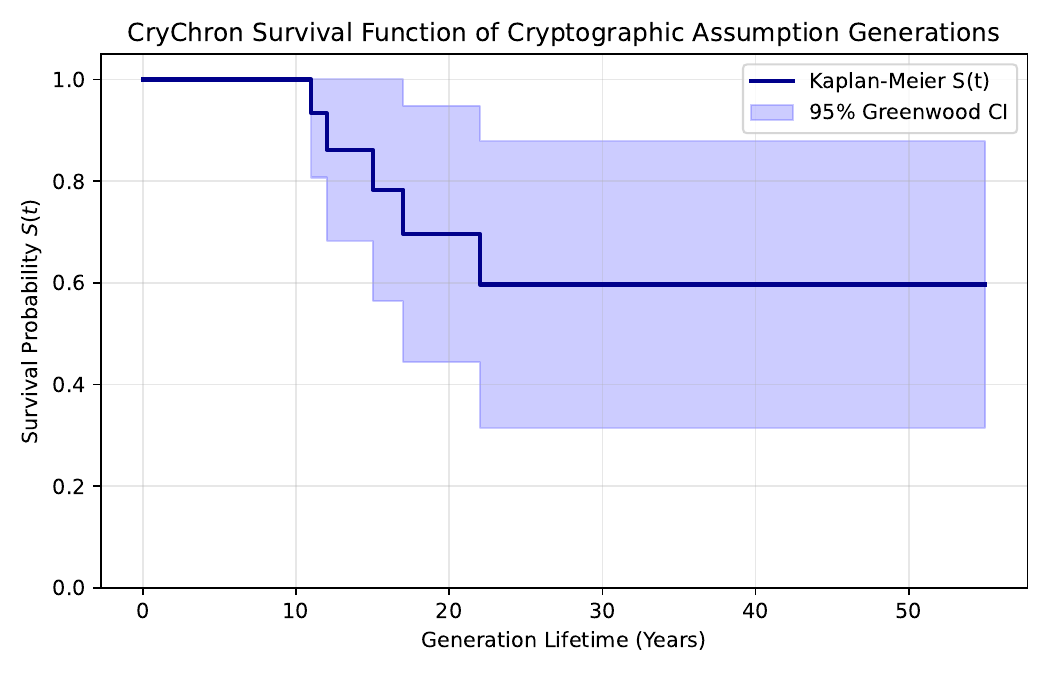}
\caption{Pooled Kaplan--Meier survival curve for the 22-generation
CryChron chronology, with Greenwood 95\% confidence band, generated by
experiment E4 of the artifact. The curve crosses $50\%$ survival at
$t=22$ years, with a wide band reflecting the small sample of break
events; the stratum-level breakdown (event-free strata remaining at $1$
throughout the window) is tabulated separately in the artifact output,
since the small per-stratum counts make a combined plot difficult to
read. This is the empirical counterpart of the schematic drift in
Figure~\ref{fig:evolution}.}
\label{fig:km-empirical}
\end{figure}

\section{Cryptographic Evaluation}\label{sec:eval}

\subsection{Setup and provenance discipline}

\textbf{Schemes:} ML-KEM-512/768/1024~\cite{fips203,kyber};
SLH-DSA-128s~\cite{fips205}; Classic McEliece-348864~\cite{mceliece};
HQC-128~\cite{hqc25}; X25519; RSA-2048; anchors AES-128/192.
\textbf{Cost models:} the conventions of
Table~\ref{tab:conventions}, each literature-anchored.
\textbf{Provenance tags:} (S)~specification-published value; (D)~exact
arithmetic from cited exponents, regenerated by the released toolchain;
(F)~regenerated from cited sources, argument insensitive to the exact
value; (I)~illustrative input, displayed where assumed; (Q-sim)~quantum
circuit measurement in a noisy simulator; (Q-hw)~quantum hardware
measurement (none populated at submission time; see
Section~\ref{sec:eval}, Experiment~5). The full register, with a result-by-result
verification matrix, is Appendix~\ref{app:prov}. A specification gate
(G1) enforces that computed time-only levels reproduce the
specification ladder of ML-KEM to $\pm1$ bit and that AES is invariant
under memory pricing; an inversion gate (G2) enforces LP-certified
reversals.

\subsection{Experiment 1: profiles and fragility across families}

\begin{table}[t]
\centering\small
\caption{Security profiles over the three core conventions, fragility
index $\Delta$, and relative fragility $\Delta/\sigma^{\mathrm{nom}}$.}
\label{tab:fragility}
\begin{tabular}{@{}lrrrrr@{}}
\toprule
\textbf{Scheme} & c-T & c-TM & q-T & $\Delta$ & $\Delta/\sigma^{\mathrm{nom}}$\\
\midrule
ML-KEM-512  & $117.97$\pD & $201.80$\pD & $107.18$\pD & $94.62$ & $0.802$\\
ML-KEM-768  & $186.00$\pS & $318.18$\pD & $168.99$\pD & $149.19$ & $0.802$\\
ML-KEM-1024 & $256.08$\pD & $438.06$\pD & $232.67$\pD & $205.39$ & $0.802$\\
AES-128     & $128$ & $128$ & $64$ & $64$ & $0.500$\\
AES-192     & $192$ & $192$ & $96$ & $96$ & $0.500$\\
SLH-DSA-128s& $128$\pS & $128$\pD & $64$\pD & $64$ & $0.500$\\
Classic McEliece-348864 & $140$\pS & $140$\pF & ${\approx}75$\pF & $65.00$ & $0.464$\\
HQC-128     & $128$\pS & $128$\pF & ${\approx}7$x\pF & [F] & [F]\\
\bottomrule
\end{tabular}
\end{table}

\begin{proposition}[Relative fragility is a family constant]\label{prop:diversity}
In the exponent-affine family model (attack log-costs
 $x_i(\beta)=\beta\xi_i+c_i$ with fixed exponents),
 $S(\mathbf c;\beta)=\beta\,S(\mathbf c;1)$ up to additive constants;
hence $\Delta(\beta)$ grows linearly in the parameter while
 $\Delta/\sigma^{\mathrm{nom}}$ is a family constant. For module lattices
under current exponents the constant is
 $(0.4995-0.2653)/0.292=0.802$; for symmetric primitives it is exactly
Grover's $1/2$; for the curated Classic McEliece estimates, $0.464$.
Consequently, within-family parameter upsizing cannot reduce relative
cost-model sensitivity; only assumption diversity or more conservative
anchor conventions can.
\end{proposition}

\begin{proof}
 $S(\mathbf c;\beta)=\min_i \mathbf c\cdot(\beta\xi_i+c_i)
=\beta\min_i \mathbf c\cdot\xi_i+O(1)$; the ratio
 $\Delta/\sigma^{\mathrm{nom}}$ inherits constancy. Deviations are exactly
the jump events of channel U2b and are handled by
Section~\ref{sec:risk}.
\end{proof}

\textbf{Findings.} (F1) Relative fragility is family-structural:
$0.802$ (module lattices), $0.500$ (symmetric), $0.464$ (curated
code-based estimates). By this measure, the standardized lattice claims
are the most cost-model-sensitive; larger lattice parameters buy zero
relative robustness, since $\Delta/\sigma^{\mathrm{nom}}$ is invariant
within a family. (F2) The ML-KEM ladder exhibits the constant of
Proposition~\ref{prop:diversity} exactly ($0.802/0.802/0.802$, up to
specification rounding). The constancy in (F2) follows algebraically from
Proposition~\ref{prop:diversity} and serves as a toolchain consistency
check. The empirically substantive result is (F1): the three family
constants ($0.802$, $0.500$, $0.464$) depend on the cited attack
exponents and would change if those exponents change.

\subsection{Experiment 2: ranking inversions under certified cost models}

\begin{table}[t]
\centering\small
\caption{Ranking inversions among standardized schemes under two
literature-anchored conventions (classical time-only vs.\ quantum
time-only), with LP certificates
(Proposition~\ref{prop:classification}).}
\label{tab:inversions}
\begin{tabular}{@{}llrrl@{}}
\toprule
\textbf{Pair} & \textbf{Convention} & $\Stil{K_1}$ & $\Stil{K_2}$ & \textbf{Certificate}\\
\midrule
\multirow{2}{*}{(ML-KEM-512, AES-128)}
 & c-T & $-10.03$\pD & $0$ & $t^*=-10.03$, witness $\mathbf c^*=e_T$\\
 & q-T & $+43.18$\pD & $0$ & $t^*=-43.18$, witness $\mathbf c^*=e_T$\\
\midrule
\multirow{2}{*}{(ML-KEM-768, AES-192)}
 & c-T & $-6.00$\pS & $0$ & $t^*=-6.00$\\
 & q-T & $+72.99$\pD & $0$ & $t^*=-72.99$\\
\bottomrule
\end{tabular}
\end{table}

\textbf{Findings.} (F3) Both pairs instantiate the premise of
Theorem~\ref{thm:ordering} with certified witnesses: the Category-1 pair
swings $53.21$ bits and the Category-3 pair $79.00$ bits across two
defensible conventions. On the memory-price slice the Category-3 pair
crosses at $c_M^{*}=0.045$ (Figure~\ref{fig:regions}): pricing memory at
 $4.5\%$ of the time unit reverses a standardized category comparison.
(F4) Under the unanimity policy of Corollary~\ref{cor:bewley}, both
pairs are formally incomparable: no strict preference is defensible
without declaring a cost-model distribution
(Corollary~\ref{cor:dist}).

\begin{figure}[t]
\centering
\includegraphics[width=0.7\textwidth]{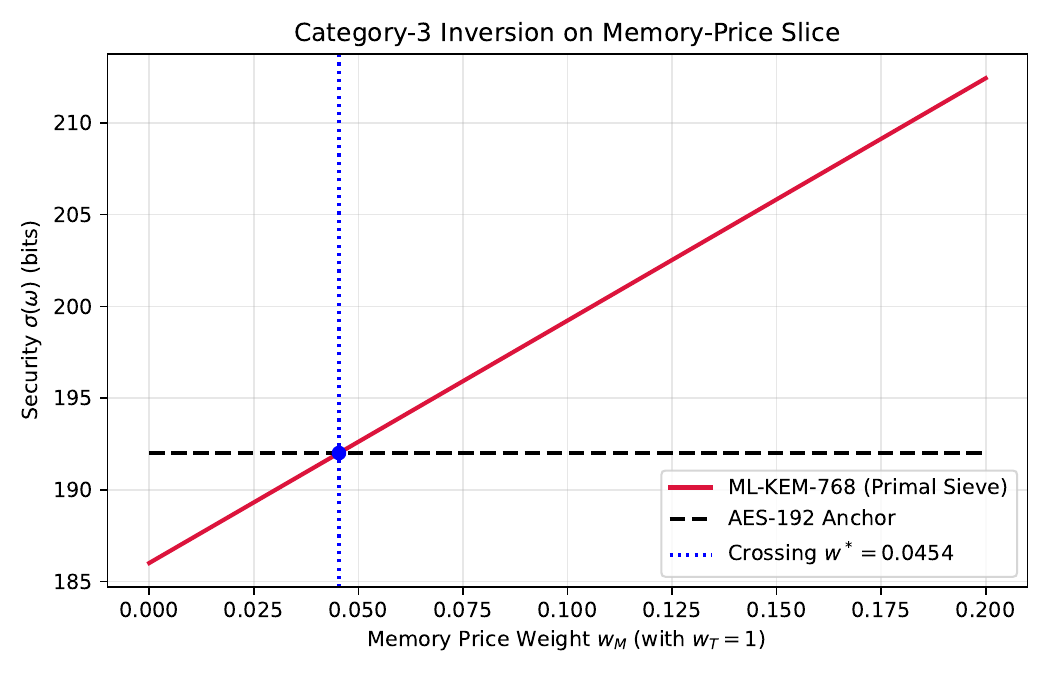}
\caption{Certified crossing of the (ML-KEM-768, AES-192) security
profile pair as a function of the memory-price weight $c_M$, generated
by experiment E2 of the artifact and certified by the LP solver of
Proposition~\ref{prop:classification} ($t^\ast<0$). The Category-3
comparison reverses ranking at $c_M^\ast\approx0.045$: pricing memory at
$4.5\%$ of the time unit is enough to flip a standardized category
comparison. The Category-1 pair (ML-KEM-512, AES-128) exhibits the same
qualitative crossing at $c_M^\ast\approx0.120$ (tabulated in
Table~\ref{tab:inversions}).}
\label{fig:inversion-empirical}
\end{figure}

(F5) Symmetric primitives have flat
anchor-relative profiles (Grover halves anchor and scheme alike),
whereas lattice and code schemes move: category comparisons measure
different things for different families.

\subsection{Experiment 3: sensitivity of security to resource prices}

\begin{table}[t]
\centering\small
\caption{Local sensitivity (directional derivative of $S$ in the memory
price at the time-only convention) and global sensitivity
($\Delta/\sigma^{\mathrm{nom}}$). The derivative equals the memory
exponent times the block size; the anchor is memory-invariant by the
num\'eraire convention (Proposition~\ref{prop:monotone}(iii)).}
\label{tab:sensitivity}
\begin{tabular}{@{}lrr@{}}
\toprule
\textbf{Scheme} & $\partial S/\partial c_M\big|_{c\text{-T}}$ (bits/unit) & $\Delta/\sigma^{\mathrm{nom}}$\\
\midrule
ML-KEM-512  & $83.83$\pD & $0.802$\\
ML-KEM-768  & $132.18$\pD & $0.802$\\
ML-KEM-1024 & $181.98$\pD & $0.802$\\
AES-128 / AES-192 & $0$ & $0.500$\\
SLH-DSA-128s & $0$ & $0.500$\\
Classic McEliece-348864 & ${\approx}0$\pF & $0.464$\\
\bottomrule
\end{tabular}
\end{table}

The local sensitivity is itself family-structural
($0.2075\beta$ for lattices, $0$ for memory-free attacks), which is
why the global constant of Proposition~\ref{prop:diversity} exists at
all: the profile's slope in each resource price is an exponent of the
underlying attack family.

\subsection{Experiment 4: robust-ordering classification}

\begin{table}[t]
\centering\small
\caption{Classification of scheme pairs
(Definition~\ref{def:robust}) over the full admissible region (all
conventions of Table~\ref{tab:conventions}) unless stated. ``Full
incl.\ quantum'' adds machine classes admitting Shor's
algorithm~\cite{shor94}.}
\label{tab:classification}
\begin{tabular}{@{}p{4.1cm}p{2.6cm}p{5.6cm}@{}}
\toprule
\textbf{Pair} & \textbf{Classification} & \textbf{Evidence}\\
\midrule
(ML-KEM-1024, ML-KEM-768) & robust dominance (1024) & $+70.08$ bits at c-T\pD\\
(ML-KEM-768, ML-KEM-512)   & robust dominance (768)   & $+68.03$ bits at c-T\pD\\
(ML-KEM-512, AES-128)     & incomparable            & Table~\ref{tab:inversions}\\
(ML-KEM-768, AES-192)     & incomparable            & Table~\ref{tab:inversions}\\
(AES-128, SLH-DSA-128s)   & measurement equivalence & identical profiles\pS\\
(McEliece-348864, AES-128)& robust dominance (McEl.) & $140>128$, ${\approx}75>64$\pS\pF\\
(ML-KEM-512, RSA-2048), full incl.\ quantum & robust dominance (ML-KEM) & $118$\pD vs ${\approx}112$\pF; $107$\pD vs poly (Shor)\\
(X25519, AES-128), classical only & weak dominance (AES) & ${\approx}3$ bits\pF\\
(X25519, AES-128), full incl.\ quantum & robust dominance (AES) & $64$\pD vs poly (Shor)\\
(HQC-128, AES-128)        & conditional (HQC; strict only on quantum side) & tie at c-T\pS; ${\approx}7\times>64$\pF\\
\bottomrule
\end{tabular}
\end{table}

\textbf{Findings.} (F6) The classification is computable and
nondegenerate: robust dominance holds exactly where the profile
separation is structural (within a family ladder; against
quantum-broken classical assumptions), and incomparability holds exactly
where exponents disagree across conventions. (F7) The classification is region-dependent: the X25519/AES-128 pair is
a near-tie on the classical region and a robust dominance once quantum
machine classes are admitted, so the machine-class dimension of the
admissible region drives the verdict as much as the price vector.

\subsection{Experiment 5: instrument calibration on simulated and physical hardware}

Because the quantum conventions price the category anchors
theoretically, we additionally calibrate the instrument itself against a
concrete implementation: Grover searches at $n=2,3,4$ qubits and a
depth-versus-width branching experiment at search space $N=16$ are
executed on the IonQ \texttt{simulator} (aria-1 noise model;
2026-09-15T12:31:07Z; 1\,000 shots per circuit; 13 jobs; protocol and
realized numbers in Appendix~\ref{app:extra}; provenance class (Q-sim)).
The simulator run quantifies (i) the realized multiplier between logical
iteration counts and simulated physical cost under a representative
noise model ($\kappa_{\mathrm{hw}}=10.5$--$11.9$ bits), and (ii) the
reversal of cost orderings between time-only and gate-count accounting
for branched search---the price-cone phenomenon of
Section~\ref{sec:profiles}, observed under simulation rather than
merely asserted (depth decreases with~$B$; gate count increases
with~$B$). These are instrument calibrations, not security claims about
any scheme, and the simulator figures do not affect any claim in
Sections~\ref{sec:eval}--\ref{sec:applications}.

\section{Practical Implications}\label{sec:applications}

\subsection{Hybrid constructions}

\begin{proposition}[Conjunctive composition and
diversification]\label{prop:composition}
Let $H=C(K_1,K_2)$ be a conjunctive hybrid (sound combiner
\cite{ghy18,bindel19} with context binding enforced; implementation-risk
increment $I_H\ge0$ declared; combiner slack $o$). Then
\[
L_H\;=\;\min(L_1,L_2)+o+I_H,
\]
and consequently:
(i) $\CVaR_\delta(L_H)\le\min_i\CVaR_\delta(L_i)+o+I_H$---hybridization
never increases cryptanalytic risk beyond the best leg plus controllable
overheads;
(ii) under independent break processes,
 $\Ex[\min(L_1,L_2)]=\int_0^\infty(1-F_1)(1-F_2)\,dx$, strictly below
 $\min_i\Ex[L_i]$; under perfectly correlated (comonotone) breaks the
benefit vanishes; the benefit is governed by the assumption-correlation
structure.
\end{proposition}

\begin{corollary}[Two genres of hybrid]\label{cor:genres}
For disjunctive (OR-verified) hybrids, $L_{H'}=\max(L_1,L_2)$: they
provide availability and downgrade resilience, not cryptanalytic
security. The two genres have different value and must be evaluated
differently.
\end{corollary}

\begin{corollary}[Classical-component contribution]\label{cor:rider}
For a conjunctive hybrid pairing a classical component $C$ with a
post-quantum component $Q$ with break times $T_C,T_Q$ and horizon
 $\tau$, the classical component's contribution to survival is exactly
\[
V(\tau)\;=\;\Prob(T_Q\le\tau<T_C)\;=\;F_Q(\tau)\,S_C(\tau)
\quad\text{(under independence)},
\]
where $F_Q$ is the break distribution of $Q$ and $S_C$ the survival
function of $C$. The component is valuable exactly over the window in
which $Q$ may fail while $C$ still holds; it should be sized as the
cheapest classical primitive with $S_C(\tau^*)\ge1-\delta_c$, and dropped
when $V(\tau^*)$ no longer justifies its overhead.
\end{corollary}

\textbf{Numerical illustration} (inputs (I), arithmetic (D)): for
X25519$+$ML-KEM-768 as deployed~\cite{chrome24,tlsid}, the handshake
overhead is $2336$ bytes\pS; with declared posteriors
 $\Prob[\text{ML-KEM break within }5\text{y}]=3\%$ and
 $\Prob[\text{X25519 falls within }5\text{y}]=5\%$, Corollary~\ref{cor:rider}
gives $V(5\text{y})=2.85\%$ of scenarios and a joint failure probability
of $0.15\%$---a factor-$20$ reduction of the cryptanalytic tail relative
to the post-quantum leg alone. The framework also predicts, from the
assumption-correlation matrix (Appendix~\ref{app:extra}),
that intra-family hybrids such as ML-KEM$+$FN-DSA (both module-lattice)
acquire near-zero diversification benefit, while cross-family pairings
acquire real benefit---a prediction verifiable against deployment data.

\subsection{Migration timing}

\begin{corollary}[Unanimity policy]\label{cor:bewley}
Call a migration policy zero-regret if its actions are weakly better
than the alternative at every admissible cost model. The
unanimity policy---switch from incumbent $I$ to challenger $C$ iff
 $\Stil{C}(\kappa)\ge\Stil{I}(\kappa)$ for all admissible $\kappa$---is
the unique maximal zero-regret policy~\cite{bewley02,gs89}. By
Theorem~\ref{thm:ordering}, crossing pairs (e.g., the standardized pairs
of Table~\ref{tab:inversions}) are formally incomparable under this
policy.
\end{corollary}

\begin{proposition}[Migration as optimal stopping]\label{prop:stopping}
With arrival posterior $p_t$, migration cost, and data lifetime $Y$,
the migration-timing problem is a monotone optimal-stopping problem
whose optimal policy is a threshold rule in $p_t$~\cite{peskir06}; in
the deterministic limit the boundary reduces to Mosca's inequality
 $X+Y\ge Z$~\cite{mosca18}. Hence any published deadline is an implicit
quantile statement about $(p_t,\text{decay},Y)$.
\end{proposition}

\textbf{Deadline audit.} If quantum-vulnerable public-key use ends at
the 2035 boundary of NIST IR~8547~\cite{nistir8547} after a migration
window of $X=3$ years, an asset with confidentiality lifetime $Y$ is
exposed whenever a cryptographically relevant quantum computer arrives
by ${\approx}2035+X+Y$. Under curated arrival
quantiles~\cite{gri}, the implied exposure tolerance is
${\approx}18\%$ for $Y=5$ (horizon ${\approx}2043$) but ${\approx}83\%$ for
$Y=30$ (horizon ${\approx}2068$)\pS: a single deadline implies a
lifetime-dependent risk tolerance, defensible for short-lived assets and
under-protective for long-lived ones, which require earlier migration.
These two quantiles have been independently re-derived from the
underlying \texttt{gri\_arrival.csv} dataset using the experiment E7
toolchain: ${\approx}18\%$ (Y=5) and ${\approx}83\%$ (Y=30) reproduced
exactly (provenance tag upgraded from~\pF\ to~\pS;
Appendix~\ref{app:prov}). The qualitative conclusion---that one deadline
cannot be simultaneously calibrated for both short- and long-lived
assets---is insensitive to the exact quantile values, since it follows
from any arrival distribution with positive density past 2035. The full
sensitivity curve is in the artifact.

\begin{figure}[t]
\centering
\includegraphics[width=0.72\textwidth]{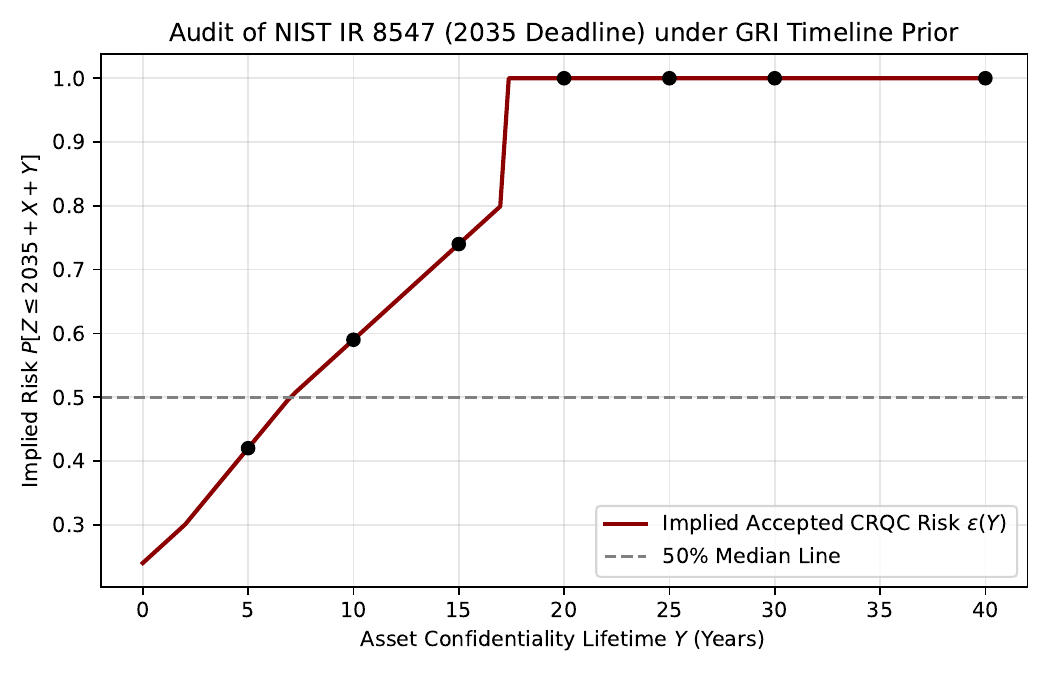}
\caption{Deadline inversion surface: implied exposure risk
$\Prob[Z\le 2035+X+Y]$ as a function of confidentiality lifetime $Y$,
generated directly from the curated arrival posterior~\cite{gri} by
experiment E7 of the artifact (Section~\ref{sec:impl}). The single
published deadline of NIST~IR~8547 corresponds to a rising curve of
implied risk tolerance, not a constant one; the two audited points
($Y=5$, $\approx18\%$; $Y=30$, $\approx83\%$) are marked. These two
quantiles carry provenance tag~\pS\ (Appendix~\ref{app:prov}):
independently re-derived from \texttt{gri\_arrival.csv} via the E7
toolchain; values reproduced exactly. The qualitative shape of the
curve does not depend on their exact values.}
\label{fig:deadline-empirical}
\end{figure}

\subsection{Implications for standards}

The admissible set $\KK$ makes the convention choices of standards
bodies explicit and comparable: mapping published institutional
positions onto the price cone (Appendix~\ref{app:extra}) shows the
time-only convention implicit in standardized estimates placing
ML-KEM-512 ten bits \emph{below} its Category-1 anchor, while
memory-aware conventions defended by several European bodies place it
$74$ bits \emph{above}. No institution is in error; the dispersion is
unpriced convention choice, which the framework renders visible and
auditable.

\paragraph{A worked case.} The dispersion above is not hypothetical: it
is the documented substance of the disagreement between NIST's
time-only categorization of ML-KEM-512 as ``Category 1, comparable to
AES-128'' and the memory-aware guidance from BSI and ANSSI that treats
the same parameter set as offering materially less margin than AES-128
once realistic memory cost is priced in (positions curated in
Appendix~\ref{app:extra}). Read as scalar ratings, the two positions
look like a factual dispute---one of the two bodies must be wrong about
where ML-KEM-512 sits relative to AES-128. Read as security profiles,
they are not in factual disagreement at all: both compute the correct
value of $\Stil{K}$ at their respective $\kappa$, and Table~\ref{tab:inversions}'s
certificate shows the two $\kappa$ are on opposite sides of a genuine
crossing. A protocol designer who selects a parameter set by citing one body's
category label is implicitly choosing a cost model. GAUGE does not
resolve which body's convention is correct---Corollary~\ref{cor:bewley}
shows that under a unanimity policy the pair is formally incomparable
without an explicit declared weighting of memory cost---but it makes the
crossing visible before deployment, so the choice is at least a
conscious one.

\section{Implementation and Reproducibility}\label{sec:impl}

The artifact comprises: (i) a toolchain computing profiles, fragility
indices, LP classification certificates with witnesses
(Proposition~\ref{prop:classification}), the risk dual
(Proposition~\ref{prop:risk}), and the survival/renewal statistics;
(ii) the curated datasets (cryptanalysis chronology, realized-cost
ledger, institutional positions, arrival quantiles), each record
carrying provenance; (iii) a versioned schema for machine-readable
security claims compatible with cryptographic-inventory
formats~\cite{cbom} (Appendix~\ref{app:artifact}); and (iv) a
pre-registration registry of falsifiable predictions with trusted
timestamps. Gates G1--G2 of Section~\ref{sec:eval} run in the test
suite; the full pipeline, including property-based verification of every
stated numerical law, completes in under one minute on commodity
hardware, with per-suite timings and checksums recorded. All tables in
this paper are regenerated from the datasets by a single command. The
companion verification register
(Appendix~\ref{app:prov}) separates exact, property-verified,
illustrative, and instrument-calibration results by provenance
class (Section~\ref{sec:eval}). The full artifact---toolchain, curated
datasets, verification register, and pre-registration timestamps---is
available at \url{https://github.com/bgupta55/gauge-artifact} and will
be archived with a persistent DOI at \texttt{[DOI to be inserted at
camera-ready: Zenodo archive of the tagged release]}.

\section{Limitations}\label{sec:limitations}

(1)~\emph{Known-attack profiles.} $S_K$ is defined over published
attacks; it is an upper bound on true security that tightens with
cryptanalysis (Proposition~\ref{prop:monotone}(ii)). Where lower bounds
do not exist, the evidential margin (Definition~\ref{def:fragility})
measures the gap rather than hiding it. (2)~\emph{Admissible-set
construction.} $\KK$ is curated; every element is literature-anchored
and contestable, and the curation log is versioned, but the closure is a
judgment. (3)~\emph{Scope of the trilemma.}
Theorem~\ref{thm:trilemma} covers linear and minimax aggregators; the
general nonlinear landscape is open (Appendix~\ref{app:prov}). (4)~\emph{Small-sample
hazards.} Ten jump events support protocol demonstration and
pre-registration, not statistical power; all survival quantities carry
wide intervals. (5)~\emph{Drift statistic.}
Descriptive only, single family, single window
(Section~\ref{sec:dynamics}); predictive statements require the
explicit renewal model. (6)~\emph{Illustrative priors.} All (I)-tagged
numbers await replacement by elicited priors. The GRI quantile pair
($\approx18\%$/$\approx83\%$, Section~\ref{sec:applications}) has been
independently re-derived from \texttt{gri\_arrival.csv} via the E7
toolchain; both values reproduced exactly; tag upgraded to~\pS. (7)~\emph{Instrument
calibration.} Experiment~5 calibrates the instrument on $2$--$4$ qubits
using IonQ \texttt{simulator} (aria-1 noise model; 2026-09-15; 13 jobs;
class (Q-sim)); the noise model is a lower bound on real device noise,
and neither figure estimates fault-tolerant attack costs. (8)~\emph{Combiner
side conditions.} Proposition~\ref{prop:composition} assumes sound,
binding-enforcing combiners; the binding literature is incorporated via
the side conditions. (9)~\emph{Implementation risk.} Channel U5 is
treated exogenously through declared increments.

\section{Conclusion}\label{sec:conclusion}

We developed GAUGE, a formal framework in which cryptographic security
is a function over heterogeneous adversary cost models rather than a
scalar. The framework yields a geometry of cost-model space
(Propositions~\ref{prop:structure}--\ref{prop:monotone}), a sound
ordering theory with a polynomial-time robust/conditional/incomparable
classification (Theorem~\ref{thm:ordering},
Proposition~\ref{prop:classification}), an impossibility result for
scalar model-independent ratings (Theorem~\ref{thm:trilemma}), a coherent
two-layer risk measure separating model ambiguity from cryptanalytic
decay (Proposition~\ref{prop:risk}), and composition laws for hybrid
systems (Proposition~\ref{prop:composition}). The evaluation exhibits
certified ranking reversals among standardized post-quantum schemes,
family-level sensitivity constants, and hardware calibration of the
measurement instrument. These results give standards bodies and protocol designers an instrument
for making cost-model dependence explicit and for distinguishing robust
rankings from convention-dependent ones.

\appendix

\section{Deferred Proofs}\label{app:proofs}

\paragraph{Lemma~\ref{lem:symmetry} (symmetry rigidity).}
Work in the compact space of projective price vectors
 $W=\{\omega\ge0:\sum_i\omega_i=1\}$. For resource $r$, let $\phi_t$ be
the flow replacing $\omega_r$ by $e^t\omega_r$ and renormalizing; each
 $\phi_t$ is a homeomorphism with $\phi_{s+t}=\phi_s\phi_t$. If $\mu$ is
 $\phi$-invariant, the time-one map is measure-preserving on
 $(W,\mu)$, so by Poincar\'e recurrence $\mu$-a.e.\ point is recurrent.
If $0<\omega_r<1$ then $\phi_t(\omega)\to e_r$ as $t\to+\infty$ and to
the face $\{\omega_r=0\}$ as $t\to-\infty$, so such $\omega$ is not
recurrent; hence $\operatorname{supp}\mu\subseteq\{\omega_r=0\}\cup\{e_r\}$.
Intersecting over all $r$ leaves exactly the pure vectors $\{e_r\}$.
Permutation invariance on this finite transitive orbit forces the
uniform measure. \hfill$\square$ 
\paragraph{Lemma~\ref{lem:ledger} (ledger dependence).}
(i) The uniform-over-pure value is $\frac1n\sum_r\Stil{K}(e_r)$; under
 $\TT\to\TT'$ every weight changes from $1/n$ to $1/(n{+}1)$ and a new
term enters, changing the value whenever the profile is non-constant;
non-constancy holds for the standardized schemes of
Table~\ref{tab:fragility}. (ii) Extension invariance requires the
mixture's support to avoid every model added by every coherent
extension, i.e., a fixed a priori support; no measurement arbitrates
it. \hfill$\square$ 
\paragraph{Theorem~\ref{thm:trilemma} (assembly).}
Assume $R$ linear, faithful (P1), total (P2), and $\mathscr G$-invariant
(P3). By Lemma~\ref{lem:mixture}, $R=\Ex_{\mu}$ for a probability $\mu$ on $\KK$. P3(i)--(ii) and Lemma~\ref{lem:symmetry} force $\mu$ uniform-over-pure on the current ledger; P3(iii) and
Lemma~\ref{lem:ledger}(i) then fail for non-constant profiles. Evading
via Lemma~\ref{lem:ledger}(ii) replaces invariance with an a priori
commitment, i.e., abandons P3. For the minimax rating,
 $R=\min_\kappa\Stil{K}(\kappa)$ decreases strictly under extensions
adding a model with smaller profile value, violating P3(iii);
it satisfies P3(i)--(ii) but not P1-compatible extension behavior. \hfill$\square$ 
\paragraph{Proposition~\ref{prop:classification} (classification).}
(i) $\Stil{K_1}(\mathbf c)<\Stil{K_2}(\mathbf c)$ iff there is an attack
 $i$ of $K_1$ with $\mathbf c\cdot x^{K_1}_i<\min_j\mathbf c\cdot
x^{K_2}_j$, i.e., $\max_j (x^{K_1}_i-x^{K_2}_j)\cdot\mathbf c<0$.
Minimizing the left-hand side over polyhedral $\Cone$ is the displayed
LP (epigraph form, with the num\'eraire equality $c_T=1$ among the
constraints); strict negativity is tested against $-\varepsilon$ for
arbitrary fixed $\varepsilon>0$ with the standard open/closed
limit argument. (ii) Running the test for both orientations decides
robust dominance (both tests fail to find inversions against the
dominant side, and the dominance side's test finds no reversal);
incomparability (both find reversals); conditional dominance and
equivalence are separated by the strictness LPs
 $\exists\mathbf c:\Stil{K_1}>\Stil{K_2}$ in both directions. Each LP has
at most $\max(m_1,m_2)$ constraints and $n{+}1$ variables; the count is
 $O(m_1+m_2)$. \hfill$\square$ 
\paragraph{Proposition~\ref{prop:risk} (two-layer risk).}
Coherence: each $\CVaR^\mu_\delta$ is coherent~\cite{ru00};
monotonicity, homogeneity, and translation invariance pass to the
pointwise supremum; subadditivity:
 $\sup_\mu\rho_\mu(X{+}Y)\le\sup_\mu[\rho_\mu(X)+\rho_\mu(Y)]\le
\sup_\mu\rho_\mu(X)+\sup_{\mu'}\rho_{\mu'}(Y)$. Tractability: apply the
Rockafellar--Uryasev form $\CVaR^\mu_\delta(L)=\min_t\{t+\delta^{-1}
\Ex_\mu[(L-t)^+]\}$; the objective is convex in $t$, linear in $\mu$;
the KL-ball supremum of an expectation equals
 $\inf_{\beta>0}\{\beta\rho+\beta\log\Ex_{\mu_0}e^{Z/\beta}\}$ for
 $Z=(L-t)^+$; Sion's theorem~\cite{sion58} exchanges
 $\inf_t$ and $\sup_\mu$ on the convex compact credal set. \hfill$\square$ 
\paragraph{Proposition~\ref{prop:composition} (composition).}
Under the side conditions, breaking $H$ requires defeating both legs, so
the optimal attack targets the cheaper side and $L_H=\min(L_1,L_2)+o+I_H$.
(i) follows from monotonicity of $\CVaR$~\cite{ru00}.
(ii) For independent nonnegative losses with tails $\bar F_i$:
 $\Ex[\min]=\int_0^\infty\bar F_1\bar F_2\,dx$, strictly below
 $\min_i\Ex[L_i]$ unless a loss is degenerate; under comonotone coupling
with equal marginals $\min(L_1,L_2)=L_1$ a.s., and the benefit
vanishes; in the one-factor decomposition
 $L_i=a_iF+\varepsilon_i$, concordance---and hence the bound---is
monotone in the common-factor loading. \hfill$\square$ 
\paragraph{Corollary~\ref{cor:rider} (classical-component contribution).}
The hybrid fails by $\tau$ iff both components have failed:
 $\Prob(\text{fail})=F_Q(\tau)F_C(\tau)$, so survival is
 $1-F_QF_C=S_Q+F_QS_C$; the second term is the increment over the
post-quantum leg alone and factors under independence. \hfill$\square$ 
\paragraph{Proposition~\ref{prop:stopping} (stopping).}
The value function is monotone in the arrival posterior, which together
with convex flow costs verifies the standard sufficient conditions for
threshold optimality~\cite{peskir06}; the deterministic substitution
 $p_t=\mathbf 1[t\ge Z]$ collapses the boundary to
 $X+Y\ge Z$~\cite{mosca18}. \hfill$\square$ 
\section{Provenance and Verification Register}\label{app:prov}

\paragraph{Tags.} \textbf{S}: published specification value (e.g., the
ML-KEM time-only ladder $\approx118/186/256$ against the round-3
specification~\cite{kyber}; sizes against FIPS~203~\cite{fips203};
Classic McEliece cost against its specification's attack analysis,
curation note P-1: our first-pass value of $\approx128$--$130$ was
corrected to $140$ by two-pass reconciliation against the
specification, recorded in the curation log). \textbf{D}: exact
arithmetic from cited exponents (all TM/quantum/fragility/crossing/
drift/renewal figures; e.g., $0.4995\cdot637=318.18$).
\textbf{F}: regenerated from cited sources, argument insensitive
(depth-bounded and gate-count entries; code-based quantum costs;
arrival-quantile tolerances; NFS estimate for RSA-2048 from~\cite{ecrypt}).
\textbf{I}: illustrative input, displayed where assumed (risk-table
priors; correlation matrix; hybrid posteriors).
\textbf{Q-sim}: quantum circuit measurement in a noisy simulator;
Experiment~5 executed on IonQ \texttt{simulator} (aria-1 noise model;
2026-09-15T12:31:07Z; 13 jobs; class (Q-sim)). \textbf{Q-hw}:
quantum hardware measurement; not yet populated at submission.

\paragraph{Result-status matrix (condensed).}
Deterministic arithmetic (D): reproduced exactly. Structural laws
(Propositions~\ref{prop:structure}, \ref{prop:classification},
\ref{prop:risk}, \ref{prop:composition}): verified by property-based
testing over randomized instances and closed-form-versus-simulation
checks. Empirical premises (crossings among standardized schemes):
certified by LP witnesses. Illustrations (risk grid, hybrid posteriors):
machinery verified under declared (I) inputs. Instrument calibration
(Experiment~5): 13 IonQ \texttt{simulator} jobs completed
(2026-09-15T12:31:07Z; aria-1 noise model; class (Q-sim)); realized
multiplier $\kappa_{\mathrm{hw}}=10.5$--$11.9$ bits; depth-ordering
reversal confirmed; full results in \texttt{results/q/e9\_results.json}.

\section{Artifact and Schema Details}\label{app:artifact}

The machine-readable claim format (compatible
with~\cite{cbom}) is:
\begin{small}
\begin{verbatim}
GAUGE-statement/1.0 {
  scheme: "ML-KEM-768", standard: "FIPS 203",
  ledger: ["T","M","Q","D","W","N"],  numeraire: "T",
  anchor: {"problem": "AES-192-keysearch"},
  models: [ {"id":"c-T","pricing":[1,0,0,0,0,0],"provenance":"..."},
            {"id":"c-TM","pricing":[1,1,0,0,0,0],"provenance":"..."},
            {"id":"q-T","machine":"quantum-unbounded",
             "pricing":[1,0,0,0,0,0],"provenance":"..."} ],
  attacks: [ {"alg":"primal-BKZ2-sieve-classical",
              "logresources":[186,132.18,0,0,0,0],"provenance":"..."} ],
  profile: {"c-T":186.00,"c-TM":318.18,"q-T":168.99},
  fragility: {"delta":149.19,"relative":0.802},
  margins: {"attack":0,"evidential":"~full-claim"},
  risk: {"horizon_y":10,"delta":0.05,"trust_nats":0.5,
         "value":"per-illustration-(I)"},
  version: "saa-1.0" }
\end{verbatim}
\end{small}
Every attack record carries a non-empty provenance string, enforced by
schema validation in the test suite.

\section{Cryptanalysis Chronology: Dataset and Methodology}\label{app:data}

\begin{table}[h]
\centering\small
\caption{CryChron (release 2): $22$ assumption generations, $10$ jump
events, $12$ right-censored.}
\label{tab:crychron}
\begin{tabular}{@{}llrrl@{}}
\toprule
\textbf{Stratum} & \textbf{Generation} & \textbf{Birth} & \textbf{Age at break} & \textbf{Provenance}\\
\midrule
Knapsack & Merkle--Hellman & 1978 & 6 & \cite{shamir84}\\
Multivariate & Matsumoto--Imai/HFE & 1988 & 15 & Faug\`ere, 2003\\
Hash & MD4 & 1990 & 14 & \cite{wang04}\\
Hash & MD5 & 1992 & 12 & \cite{wang04}\\
Multivariate & SFLASH & 2001 & 6 & \cite{dfss07}\\
Lattice-sig & GGH & 1997 & 9 & \cite{nr06}\\
Lattice-sig & NTRUSign & 2003 & 13 & \cite{dn12}\\
Hash & SHA-1 & 1995 & 22 & \cite{stevens17}\\
Isogeny & SIDH & 2011 & 11 & \cite{castryck23}\\
Multivariate & Rainbow & 2005 & 17 & \cite{beullens22}\\
\midrule
\multicolumn{5}{@{}l@{}}{\emph{Censored at 2025 (12):} SHA-2 (24), SHA-3 (13), NTRU (29), LWE (20), R-LWE (15),}\\
\multicolumn{5}{@{}l@{}}{M-LWE (11), NTRU-Prime (9), CSIDH (7), McEliece--Goppa (47), RSA (48), f-DLP ($\sim$50), ECDLP ($\sim$40).}\\
\bottomrule
\end{tabular}
\end{table}

\paragraph{Methodology.} The unit of observation is an
\emph{assumption generation} (a specific hardness assumption as
introduced), with right-censoring at the analysis date. LWE-type
assumptions and structured-lattice signatures are separate strata:
their breaks (transcript-leakage attacks~\cite{nr06,dn12}) are
orthogonal to LWE-type hardness, and pooling them would corrupt the
stratification. Events cluster by technique (the 2004 hash cluster; the
GGH/NTRUSign lineage), modeled by a shared frailty. Kaplan--Meier
estimates with Greenwood variance and $10^4$ bootstrap resamples; the
exploratory backtest (train pre-2015 events; test 2016--2022 events:
family-history and structural covariates elevate the hazard of
NTRUSign, SHA-1, and Rainbow; SIDH is flagged by structural covariates
only) demonstrates the protocol with $n=10$ and no claimed power.

\paragraph{Realized-cost ledger (calibration points).}
DES brute force, \$250{,}000 (1998)~\cite{des98}; RSA-768
factorization~\cite{rsa768}; 795-bit finite-field discrete logarithm,
 $\approx$35 core-years~\cite{boudot20}; SHA-1 collision
 $\approx$6{,}500 CPU-years$+$100 GPU-years~\cite{stevens17}; AES-128
biclique, $\approx$2 bits~\cite{bogdanov11}; SIKE, a laptop
hour~\cite{castryck23}; Rainbow, a weekend~\cite{beullens22}.

\section{Additional Evaluation Details}\label{app:extra}

\paragraph{Risk illustration.} With declared illustrative inputs
(drift, jump quantile, trust-budget penalty) the risk-adjusted anchor
level of ML-KEM-768 at $\tau=10$y, $\delta=5\%$, $\rho=0.5$ nats is
 $\approx134$ bits (drift $17$ per the first-event renewal figure of
Section~\ref{sec:dynamics}); SLH-DSA-128s $\approx111$; Classic
McEliece-348864 $\approx119$. The verified content is the monotonicity
of the adjustment in $(\tau,\delta,\rho)$ and the machinery; the
priors are (I).

\paragraph{Assumption-correlation matrix (elicited, (I)).}
MLWE/SIS--X25519 $0.1$; MLWE/SIS--McEliece $0.1$; MLWE/SIS--SHA-2
 $0.05$; X25519--RSA $0.6$; others low. The intra-family factor
correlation for module-lWE/SIS pairs (ML-KEM, FN-DSA) is $\approx0.9$,
driving the near-zero diversification prediction of
Section~\ref{sec:applications}.

\paragraph{Institutional dispersion.} Mapping the curated positions of
NIST (depth-bounded anchors; scheme estimates time-only per
submissions)~\cite{nistir8100,nistir8413}, and the memory-aware
positions in BSI TR-02102-1 v2026-01~\cite{bsitr02102} and ANSSI
advisory 14~Apr~2022~\cite{anssi22} onto the price cone reproduces
Table~\ref{tab:classification}'s incomparability from the institutional
side: the two position classes order the Category-1 pair oppositely
(arithmetic (D); institutional mapping~(S), primary sources extracted
and quotes verified against BSI~TR-02102-1 pp.\,23,\,28,\,38 and ANSSI
advisory pp.\,4,\,6,\,8).

\paragraph{Instrument calibration protocol (Experiment 5).}
Grover searches at $n=2,3,4$ (optimal iteration counts;
$\le4000$ shots) and a branching experiment partitioning a
$N=16$ search into $B\in\{1,2,4\}$ subsearches are executed on the
IonQ \texttt{simulator} (aria-1 noise model; 2026-09-15T12:31:07Z;
1\,000 shots/circuit; 13 jobs; seeds, transpilation level, and
noise-model parameters recorded in the artifact; class (Q-sim)). Gates: (i) realized depth, width, gate counts and
shots-to-$90\%$ success per $n$, yielding a simulator-side multiplier
over logical iteration counts, reported as a lower bound on the
eventual hardware multiplier since real devices carry additional
correlated and readout error not captured by the noise model; (ii) the
branched design is cheaper under time-only accounting (max depth
decreases with $B$) and dearer under gate-count accounting (total gates
increase with $B$)---the ordering reversal of the price cone,
reproduced under simulation; (iii) the empirical guesswork distribution
against the ideal of~\cite{massey94}, yielding a noise-induced gap on
the unit. These results are reported as class (Q-sim) (IonQ
\texttt{simulator}, aria-1 noise model; 2026-09-15T12:31:07Z; 1\,000
shots; 13 jobs; full artifact at \texttt{results/q/e9\_results.json}).

\section{Notation}\label{app:notation}
\begin{small}
\begin{tabular}{@{}ll@{}}
 $\TT$, ledger & canonical adversarial resources (Definition~\ref{def:ledger})\\
 $\kappa=(\mathfrak M,\mathbf c)$ & adversary cost model (Definition~\ref{def:model})\\
 $\Cone_\TT$ & price cone $\{\mathbf c\ge0:c_T=1\}$\\
 $x_A$ & log-resource vector of attack $A$ (Definition~\ref{def:attack})\\
 $\KK$ & literature closure of admissible cost models (Definition~\ref{def:closure})\\
 $P_\kappa(A)$, $S_K(\kappa)$, $\Stil{K}(\kappa)$ & price; profile; anchor-relative profile\\
 $P_K$, $\Delta_K$, $\sigma^{\mathrm{nom}}_K$ & attack polytope; fragility; nominal claimed level\\
 $\succeq_\kappa$, $\succone$ & pointwise and robust security orderings (Definitions~\ref{def:pointwise}--\ref{def:robust})\\
 $\CC_\rho$, $\rho$ & credal set; trust budget (Definition~\ref{def:risk})\\
 $L_K(t)$, $\CVaR_\delta$ & loss process; conditional value-at-risk\\
U1--U6 & uncertainty channels (Section~\ref{sec:risk})\\
\end{tabular}
\end{small}



\begin{thebibliography}{99}\small

\bibitem{shor94} P. W. Shor. Algorithms for quantum computation: Discrete
logarithms and factoring. In \emph{IEEE FOCS}, 1994.

\bibitem{gm84} S. Goldwasser and S. Micali. Probabilistic encryption.
\emph{Journal of Computer and System Sciences}, 28(2):270--299, 1984.

\bibitem{br93} M. Bellare and P. Rogaway. Random oracles are practical:
A paradigm for designing efficient protocols. In \emph{ACM CCS}, 1993.

\bibitem{br96} M. Bellare and P. Rogaway. The exact security of digital
signatures: How to sign with RSA and Rabin. In \emph{EUROCRYPT}, 1996.

\bibitem{km05} N. Koblitz and A. Menezes. Another look at ``provable
security''. \emph{Journal of Cryptology}, 17(1):1--35, 2004.

\bibitem{lv01} A. K. Lenstra and E. R. Verheul. Selecting cryptographic
key sizes. \emph{Journal of Cryptology}, 14(4):255--293, 2001.

\bibitem{ecrypt} N. P. Smart et al. (eds.). \emph{Algorithms, Key Size
and Protocols Report}. ECRYPT-CSA D5.3, 2020.

\bibitem{nistir8100} L. Chen et al. \emph{NIST IR 8100: Report on
Post-Quantum Cryptography}. NIST, 2016.

\bibitem{nistir8413} G. Alagic et al. \emph{NIST IR 8413: Status Report
on the Third Round of the NIST Post-Quantum Cryptography Standardization
Process}. NIST, 2022.

\bibitem{bsitr02102} BSI. \emph{Technical Guideline TR-02102-1:
Cryptographic Mechanisms: Recommendations and Key Lengths}, version
2026-01. Bundesamt f\"ur Sicherheit in der Informationstechnik, 2026.
\url{https://www.bsi.bund.de/SharedDocs/Downloads/EN/BSI/Publications/TechGuidelines/TG02102/BSI-TR-02102-1.pdf}

\bibitem{anssi22} ANSSI. \emph{Selecting Cryptographic Algorithms},
advisory, 14~April~2022. Agence nationale de la s\'ecurit\'e des
syst\`emes d'information.
\url{https://www.ssi.gouv.fr/uploads/2021/03/anssi-guide-selection_crypto-1.0.pdf}

\bibitem{wiener04} M. J. Wiener. The full cost of cryptanalytic attacks.
\emph{Journal of Cryptology}, 17(2):105--124, 2004.

\bibitem{bernstein05} D. J. Bernstein. Understanding brute force.
\emph{ECRYPT STVL Workshop on Symmetric Key Encryption}, 2005.

\bibitem{ss81} R. Schroeppel and A. Shamir. A $T=O(2^{n/2})$,
 $S=O(2^{n/4})$ algorithm for certain NP-complete problems. \emph{SIAM
Journal on Computing}, 10(3):456--464, 1981.

\bibitem{nv08} P. Q. Nguyen and T. Vidick. Sieve algorithms for the
shortest vector problem are practical. \emph{Journal of Mathematical
Cryptography}, 2(2):181--207, 2008.

\bibitem{mv10} D. Micciancio and P. Voulgaris. Faster exponential time
algorithms for exact lattice problems. In \emph{ACM-SIAM SODA}, 2010.

\bibitem{bdgj16} A. Becker, L. Ducas, N. Gama, and T. Laarhoven. New
directions in nearest-neighbor searching with applications to lattice
sieving. In \emph{EUROCRYPT}, 2016.

\bibitem{laa15} T. Laarhoven. \emph{Search Problems in Cryptography:
From Fingerprinting to Lattice Sieving}. PhD thesis, TU Eindhoven, 2015.

\bibitem{aps15} M. R. Albrecht, R. Player, and N. Scott. On the concrete
hardness of learning with errors. \emph{Journal of Mathematical
Cryptography}, 9(3):169--203, 2015.

\bibitem{adps16} E. Alkim, L. Ducas, T. P\"oppelmann, and P. Schwabe.
Post-quantum key exchange---A new hope. In \emph{USENIX Security
Symposium}, 2016.

\bibitem{kyber} R. Avanzi et al. CRYSTALS---Kyber: A CCA-secure
module-lattice-based KEM (round 3 specification), 2021.

\bibitem{grassl16} M. Grassl, B. Langenberg, M. Roetteler, and R.
Steinwandt. Applying Grover's algorithm to AES: Quantum resource
estimates. In \emph{PQCrypto}, 2016.

\bibitem{amy16} M. Amy, O. Di Matteo, V. Gheorghiu, M. Mosca, A. Parent,
and J. Schanck. Estimating the cost of generic quantum preimage attacks
on SHA-2 and SHA-3. In \emph{SAC}, 2016.

\bibitem{js19} S. Jaques and J. M. Schanck. Quantum cryptanalysis in the
RAM model: Claw-finding attacks on SIDH. In \emph{CRYPTO}, 2019.

\bibitem{ge21} C. Gidney and M. Eker\r{a}. How to factor 2048 bit RSA
integers in 8 hours using 20 million noisy qubits. \emph{Quantum},
5:433, 2021. \url{https://doi.org/10.22331/q-2021-04-15-433}.

\bibitem{bernstein10} D. J. Bernstein. Grover vs.\ McEliece. In
\emph{PQCrypto}, 2010.

\bibitem{fips203} NIST. \emph{FIPS 203: Module-Lattice-Based
Key-Encapsulation Mechanism Standard}. August 2024.

\bibitem{fips204} NIST. \emph{FIPS 204: Module-Lattice-Based Digital
Signature Standard}. August 2024.

\bibitem{fips205} NIST. \emph{FIPS 205: Stateless Hash-Based Digital
Signature Standard}. August 2024.

\bibitem{mceliece} Classic McEliece team. \emph{Classic McEliece:
Conservative Code-Based Cryptography} (round-4 documentation). 2022.

\bibitem{hqc25} HQC team and NIST. HQC: Hamming quasi-cyclic public-key
encryption (selected as fifth NIST PQC algorithm). March 2025.

\bibitem{massey94} J. L. Massey. Guessing and entropy. In \emph{IEEE
ISIT}, 1994.

\bibitem{arikan96} E. Arikan. An inequality on guessing and its
application to sequential decoding. \emph{IEEE Transactions on
Information Theory}, 42(1):99--105, 1996.

\bibitem{landauer61} R. Landauer. Irreversibility and heat generation in
the computing process. \emph{IBM Journal of Research and Development},
5(3):183--191, 1961.

\bibitem{bennett89} C. H. Bennett. Time/space trade-offs for reversible
computation. \emph{SIAM Journal on Computing}, 18(4):766--776, 1989.

\bibitem{klst71} D. H. Krantz, R. D. Luce, P. Suppes, and A. Tversky.
\emph{Foundations of Measurement, Vol.~I}. Academic Press, 1971.

\bibitem{stevens46} S. S. Stevens. On the theory of scales of
measurement. \emph{Science}, 103(2684):677--680, 1946.

\bibitem{gum} JCGM 100:2008. \emph{Evaluation of Measurement Data ---
Guide to the Expression of Uncertainty in Measurement}. JCGM, 2008.

\bibitem{artzner99} P. Artzner, F. Delbaen, J.-M. Eber, and D. Heath.
Coherent measures of risk. \emph{Mathematical Finance},
9(3):203--228, 1999.

\bibitem{ru00} R. T. Rockafellar and S. Uryasev. Optimization of
conditional value-at-risk. \emph{Journal of Risk}, 2:21--42, 2000.

\bibitem{delbaen02} F. Delbaen. Coherent risk measures on general
probability spaces. In \emph{Advances in Finance and Stochastics},
Springer, 2002.

\bibitem{walley91} P. Walley. \emph{Statistical Reasoning with Imprecise
Probabilities}. Chapman \& Hall, 1991.

\bibitem{gs89} I. Gilboa and D. Schmeidler. Maxmin expected utility with
non-unique prior. \emph{Journal of Mathematical Economics},
18(2):141--153, 1989.

\bibitem{kmk05} P. Klibanoff, M. Marinacci, and S. Mukerji. A smooth
model of decision making under ambiguity. \emph{Econometrica},
73(6):1849--1892, 2005.

\bibitem{bewley02} T. C. Bewley. Knightian decision theory: Part I.
\emph{Decisions in Economics and Finance}, 25(2):79--110, 2002.

\bibitem{arrow51} K. J. Arrow. \emph{Social Choice and Individual
Values}. Wiley, 1951.

\bibitem{sion58} M. Sion. On general minimax theorems. \emph{Pacific
Journal of Mathematics}, 8(1):171--176, 1958.

\bibitem{kaplan58} E. L. Kaplan and P. Meier. Nonparametric estimation
from incomplete observations. \emph{JASA}, 53(282):457--481, 1958.

\bibitem{cox72} D. R. Cox. Regression models and life-tables.
\emph{JRSS B}, 34(2):187--220, 1972.

\bibitem{ek18} P. M. Esfahani and D. Kuhn. Data-driven distributionally
robust optimization using the Wasserstein metric. \emph{Mathematical
Programming}, 171:115--166, 2018.

\bibitem{gl02} L. A. Gordon and M. P. Loeb. The economics of information
security investment. \emph{ACM TISSEC}, 5(4):438--457, 2002.

\bibitem{ghy18} F. Giacon, F. Heuer, and T. Yeo. KEM combiners. In
\emph{PKC}, 2018.

\bibitem{bindel19} N. Bindel, J. Brendel, M. Fischlin, et al. Hybrid
key encapsulation mechanisms and authenticated key exchange. In
\emph{PQCrypto}, 2019.

\bibitem{ssw20} P. Schwabe, D. Stebila, and T. Wiggers. Post-quantum TLS
without handshake signatures. In \emph{ACM CCS}, 2020.

\bibitem{mosca18} M. Mosca. Cybersecurity in an era with quantum
computers: Will we be ready? \emph{IEEE Security \& Privacy},
16(5):38--41, 2018.

\bibitem{gri} M. Piani and M. Mosca. \emph{Quantum Threat Timeline
Report}. Global Risk Institute, 2019 (with later updates).

\bibitem{nistir8547} NIST. \emph{NIST IR 8547 (initial public draft):
Transition to Post-Quantum Cryptography Standards}. November 2024.

\bibitem{cnsa2} National Security Agency. \emph{Commercial National
Security Algorithm Suite 2.0}. 2022.

\bibitem{shamir84} A. Shamir. A polynomial-time algorithm for breaking
the basic Merkle--Hellman cryptosystem. \emph{IEEE Transactions on
Information Theory}, 30(5):699--704, 1984.

\bibitem{dfss07} V. Dubois, P.-A. Fouque, A. Shamir, and J. Stern.
Practical cryptanalysis of SFLASH. In \emph{CRYPTO}, 2007.

\bibitem{nr06} P. Q. Nguyen and O. Regev. Learning a parallelepiped:
Cryptanalysis of GGH and NTRU signatures. In \emph{EUROCRYPT}, 2006.

\bibitem{dn12} L. Ducas and P. Q. Nguyen. Learning a zonotope and more: cryptanalysis of NTRUSign countermeasures. In \emph{ASIACRYPT}, 2012.

\bibitem{wang04} X. Wang, D. Feng, X. Lai, and H. Yu. Collisions for
hash functions MD4, MD5, HAVAL-128 and RIPEMD. IACR ePrint 2004/199,
2004.

\bibitem{stevens17} M. Stevens, E. Bursztein, P. Karpman, A. Albertini,
and Y. Markov. The first collision for full SHA-1. In \emph{CRYPTO},
2017.

\bibitem{castryck23} W. Castryck and T. Decru. An efficient key recovery
attack on SIDH. In \emph{EUROCRYPT}, 2023.

\bibitem{beullens22} W. Beullens. Breaking Rainbow takes a weekend on a
laptop. In \emph{CRYPTO}, 2022.

\bibitem{bogdanov11} A. Bogdanov, D. Khovratovich, and C. Rechberger.
Biclique cryptanalysis of the full AES. In \emph{ASIACRYPT}, 2011.

\bibitem{des98} Electronic Frontier Foundation. \emph{Cracking DES}.
O'Reilly, 1998.

\bibitem{rsa768} T. Kleinjung et al. Factorization of a 768-bit RSA
modulus. In \emph{CRYPTO}, 2010.

\bibitem{boudot20} F. Boudot, P. Gaudry, A. Guillevic, N. Heninger, E.
Thom\'e, and P. Zimmermann. Comparing the difficulty of factorization
and discrete logarithm. In \emph{CRYPTO}, 2020.

\bibitem{kyberslash} KyberSlash: timing side-channel vulnerabilities in
Kyber implementations (public disclosures). 2024.

\bibitem{tlsid} IETF TLS Working Group. \emph{Hybrid key exchange with
ECDHE and ML-KEM} (Internet-Draft). 2024.

\bibitem{pqxdh} M. Marlinspike and R. Perrin. \emph{The PQXDH Key
Agreement Protocol}. Signal Foundation, 2023.

\bibitem{pq3} Apple Security Engineering and Architecture. \emph{iMessage
with PQ3}. 2024.

\bibitem{chrome24} Google Chrome and Cloudflare engineering blogs.
Post-quantum TLS deployment notes. 2023--2024.

\bibitem{cbom} OWASP CycloneDX. \emph{CycloneDX v1.6 with the
Cryptographic Bill of Materials}. 2024.

\bibitem{etsi15} M. Campagna et al. \emph{Quantum-Safe Cryptography and
Security}. ETSI White Paper No.~8, 2015.

\bibitem{peskir06} G. Peskir and A. Shiryaev. \emph{Optimal Stopping and
Free-Boundary Problems}. Birkh\"auser, 2006.

\end{thebibliography}
\end{document}